\documentclass [smallextended,nospthms] {svjour3}
\usepackage{amssymb}
\usepackage{amsmath}
\usepackage{amsfonts}
\usepackage{graphicx}
\usepackage{a4wide}
\usepackage{microtype}
\usepackage{bbm}
\usepackage{slashed}

\usepackage{amsthm}
\usepackage{mathrsfs}
\usepackage{hyperref}
\usepackage{color}

\definecolor{darkred}{rgb}{0.8,0.1,0.1}
\hypersetup{
     colorlinks=true,         % false: boxed links; true: colored links,false is default
     linkcolor=darkred,
     citecolor=blue,
}

\usepackage{comment}
\usepackage[margin=2.7cm]{geometry}
\usepackage[all,cmtip]{xy}

\usepackage{marvosym}

\theoremstyle{plain}
\newtheorem{theo}{Theorem}[section]
\newtheorem{lem}[theo]{Lemma}
\newtheorem{propo}[theo]{Proposition}
\newtheorem{cor}[theo]{Corollary}

\theoremstyle{definition}
\newtheorem{defi}[theo]{Definition}
\newtheorem{ex}[theo]{Example}

\newtheorem{rem}[theo]{Remark}

\numberwithin{equation}{section}

\def\nn{\nonumber}

\def\bbK{\mathbb{K}}
\def\bbR{\mathbb{R}}
\def\bbC{\mathbb{C}}
\def\bbN{\mathbb{N}}

\def\EE{\mathcal{E}}
\def\HH{\mathcal{H}}
\def\DD{\mathcal{D}}

\def\End{\mathrm{End}}

\def\Sol{\mathrm{Sol}}
\def\Ker{\mathrm{Ker}}
\def\Imm{\mathrm{Im}}

\def\id{\mathrm{id}}
\def\supp{\mathrm{supp}}

\def\dd{\mathrm{d}}
\def\vol{\mathrm{vol}_M^{}}
\def\vols{\mathrm{vol}_\Sigma^{}}
\def\sc{\mathrm{sc}}
\def\dim{\mathrm{dim}}

\def\1{\mathbbm{1}}
\def\g{\mathfrak{g}}
\def\w{\mathsf{w}}
\def\b{\mathsf{b}}

\newcommand{\ip}[2]{\langle #1,#2 \rangle}

\def\sk{\vspace{2mm}}

\title{%
Linear bosonic and fermionic quantum gauge theories \\ on curved spacetimes
}

\author{%
 Thomas-Paul Hack \and Alexander Schenkel
 }

\institute{II.~Institut f\"ur Theoretische Physik,~Universit\"at Hamburg,~Luruper Chaussee~149,~22761~Hamburg,~Germany. \email{thomas-paul.hack@desy.de} \and
 Fachgruppe Mathematik,~Bergische~Universit\"at~Wuppertal,~Gau\ss stra\ss e~20,~42119~Wuppertal,~Germany. \email{schenkel@math.uni-wuppertal.de}}

\date{\today}

\begin{document}

\maketitle

\begin{abstract}
We develop a general setting for the quantization of linear bosonic and fermionic field theories subject to 
local gauge invariance and show how standard examples such as linearised Yang-Mills theory and linearised general relativity
 fit into this framework. Our construction always leads to a well-defined and gauge-invariant quantum field algebra, 
 the centre and representations of this algebra, however, have to be analysed on a case-by-case basis. 
 We discuss an example of a fermionic gauge field theory where the necessary conditions for the existence of Hilbert space
  representations are not met on any spacetime. On the other hand, we prove that these conditions are met for the Rarita-Schwinger 
  gauge field in linearised pure $N=1$ supergravity on certain spacetimes, including asymptotically flat spacetimes 
  and classes of spacetimes with compact Cauchy surfaces. We also present an explicit example of a supergravity
  background on which the Rarita-Schwinger gauge field can not be consistently quantized.
\end{abstract}


%%%%%%%%%%%%%%%%%%%%%%%%%%%%%%%%%%%%%%%%%%%%%%%%%%%%%%%
%%%%%%%%%%%%%%%%%%%%%%%%%%%%%%%%%%%%%%%%%%%%%%%%%%%%%%%

\section{\label{sec:intro}Introduction}
Quantum field theory on curved spacetimes has gone through  
major developments in the last decades. Explicit models have been constructed in this framework, including the
scalar field \cite{DimockScalar}, the Dirac field \cite{DimockDirac,Sanders:2008kp,Dappiaggi:2009xj} 
and the Proca field \cite{Furlani:1999kq}. These examples have later been recast into a general 
approach to the quantization of bosonic and fermionic matter field theories on curved spacetimes \cite{Bar:2007zz,Bar:2011iu}.
On the other hand, examples of theories exhibiting a local gauge invariance have been investigated in detail,
including the Maxwell field \cite{DimockVector,Fewster:2003ey,Pfenning:2009nx,Dappiaggi:2011cj,Dappiaggi:2011zs,DHS} and linearised general relativity
on Einstein manifolds \cite{Fewster:2012bj}. The quantization of gauge field theories bears new complications, which are not present
for matter field theories. In particular, the equation of motion in a gauge field theory is not hyperbolic and thus one does not have
a well-defined Cauchy problem or Green's operators,  which are the basic structures entering the construction
of matter quantum field theories. This problem has been resolved in the examples mentioned above 
by considering only the gauge invariant content of such a theory,
i.e.~gauge invariant observables, and making use of a special gauge fixing condition. We emphasise that even though
a gauge fixing is used in this construction, the resulting algebra of observables is by definition gauge invariant.
The algebra of gauge invariant observables of a gauge field theory can
have new features compared to matter field theories.
As it has been shown in \cite{Bar:2007zz,Bar:2011iu} (see also Section \ref{sec:algebra} in the present paper) the algebra
of observables of a bosonic matter quantum field theory never has a non-trivial centre. In gauge field theories
this can in general only be guaranteed under additional assumptions on the Cauchy surface in the spacetime,
see \cite{DimockVector} for the Maxwell field and \cite{Fewster:2012bj} for linearised general relativity on Einstein manifolds.
There are examples of Cauchy surfaces such that the algebra of gauge invariant observables of the
Maxwell field has a non-trivial centre \cite{Dappiaggi:2011zs,DHS}.
Due to the theory of degenerate Weyl algebras \cite{degenerateCCR}  these centres do not pose 
mathematical problems for the quantum field theory on an individual spacetime, but
they have impact on whether or not the theory is locally covariant in the sense of \cite{Brunetti:2001dx}, 
see e.g.~\cite{Dappiaggi:2011zs,DHS}.
Furthermore, the centres are certainly of physical interest and should be understood in detail. 
We also want to mention that in addition to these results on linear
quantum gauge field
theories there has been a lot of effort in constructing perturbatively
interacting quantum gauge field theories
on curved spacetimes, see e.g.~\cite{Hollands:2007zg,Fredenhagen:2011mq} and references therein. 
In our work we restrict ourselves to linear quantum field theories, since as it will become clear later,
there are a lot of non-trivial aspects which have to be understood in detail even at the linear level.
This is in particular the case for fermionic gauge field theories. The restriction to linear theories will allow us to quantize gauge fields without introducing auxiliary fields as it happens in the BRST/BV-formalism, cf.~\cite{Hollands:2007zg,Fredenhagen:2011mq}. However, we presume that our construction for the bosonic case yields a gauge invariant algebra of quantum observables which is isomorphic to the one obtained in \cite{Hollands:2007zg,Fredenhagen:2011mq} at lowest order in perturbation theory.

The goal of the present paper is twofold: First, we aim at developing a general framework for the quantization
of linear gauge field theories. This can be seen as an extension of \cite{Bar:2007zz,Bar:2011iu} to field theories
subject to a local gauge invariance. We allow for bosonic as well as fermionic theories and provide
an axiomatic definition of a classical linear gauge field theory in terms of fibre bundles and differential operators thereon.
Our setting is general enough to cover the matter field theories of \cite{Bar:2007zz,Bar:2011iu}, which
will be promoted to gauge field theories with a trivial gauge structure, as well as the standard
examples such as linearised Yang-Mills theory and linearised general relativity on Einstein manifolds.
Even more, our general framework is sufficiently flexible to include examples of fermionic gauge field theories.
The prime example of such a theory is the gravitino field (also called Rarita-Schwinger field) in 
linearised pure $ N=1$ supergravity, which we will discuss in detail.
A further example which we will study in detail is a fermionic version of linearised Yang-Mills theory, which emerges
for example as the fermionic sector of a Yang-Mills theory modeled on a Lie supergroup.
Bosonic gauge field theories can always be quantized in terms of (possibly degenerate) Weyl algebras,
while fermionic gauge field theories bear additional complications, similar to their matter field theory counterparts
\cite{Bar:2007zz,Bar:2011iu}. The issue there is that the inner product space associated to a fermionic matter or gauge field theory
is in general indefinite, and one therefore encounters physical as well as mathematical problems. The mathematical issue
is that such indefinite inner product spaces can not be quantized with the usual CAR-representation. The physical
problem is that, even if there would exist a suitable CAR-algebra, there are negative norm states in any representation of it.
In contrast to other approaches to the quantization of gauge field theories which are based on kinematical (i.e.~still 
containing gauge degrees of freedom) representation spaces, our negative norm states would be states in the physical 
(i.e.~gauge invariant) Hilbert space and would thus pose problems for
the physical interpretation of the fermionic gauge field theory under consideration.
This brings us to the second goal of this paper, which is the investigation under which conditions the two examples
of fermionic gauge field theories give rise to positive definite inner product spaces and thus 
can be consistently quantized in terms of a CAR-representation. We prove that the fermionic generalisation of linearised
Yang-Mills theory always leads to an indefinite inner product space and thus can not be quantized on any spacetime.
This implies that the perturbative quantization of Yang-Mills theories based on Lie supergroups is, in the above mentioned
sense, inconsistent and puts strong mathematical constraints on such theories. On the other hand, our result is well in line with the spin-statistics theorem. The situation is better for the gravitino field
of linearised pure supergravity. We provide a sufficient condition for this theory to give rise to a positive inner product space,
which demands the existence of a special type of gauge transformation. For compact  Cauchy surfaces 
this condition is fulfilled if the induced (Riemannian) Dirac 
operator on the Cauchy surface has a trivial kernel.
We also consider certain non-compact Cauchy surfaces and answer the question
of positivity affirmatively. This shows that, under assumptions on the Cauchy surface, treating the 
Rarita-Schwinger field as a fermionic gauge field theory
(as it is required by supergravity) improves on well-known issues appearing in the quantization of the Rarita-Schwinger
field when treated as a matter field theory, see e.g.~\cite{Bar:2011iu,Hack:2011yv,Schenkel:2011nv}.
Introducing a mass term for the gravitino field in a gauge-invariant way requires the coupling
of matter fields to the supergravity and will be discussed elsewhere.
We also provide an example of a supergravity background on which the Rarita-Schwinger gauge 
field can not be consistently quantized via a CAR-representation.
Considering the  spacetime $M=\bbR\times\mathbb{T}^{D-1}$ -- with
$\mathbb{T}^{D-1}$ denoting the $D{-}1$-torus --
 equipped with the flat Lorentzian metric, we show that in case of the trivial spin structure
the inner product is indefinite, while for all other spin structures it is positive definite. 
A  complete classification of Cauchy surfaces and induced metrics thereon which lead to a positive inner product for the
 Rarita-Schwinger gauge field
seems to be very complicated and is beyond the scope of this work.

The outline of this paper is as follows: In Section \ref{sec:notation} we review some basic aspects of Lorentzian geometry
and differential operators on vector bundles following mainly the presentation in \cite{Bar:2007zz,Bar:2011iu}.
We then introduce our definition of classical gauge field theories in Section \ref{sec:class} and show  that
the basic examples studied in the literature fit into this framework. We conclude this section with a theorem
on properties of classical gauge field theories, which generalises the properties found in the explicit examples
to the axiomatic level. In Section \ref{sec:algebra} we study the quantization of gauge field theories and in particular
propose suitable algebras of gauge invariant observables. The question of non-degeneracy (positivity)
of bosonic (fermionic) gauge field theories is investigated in Section \ref{sec:degeneracy}. 
The Rarita-Schwinger gauge field is discussed separately in Section \ref{sec:degeneracyrarita}.
Appendix \ref{app:conventions} contains our spinor conventions.

%%%%%%%%%%%%%%%%%%%%%%%%%%%%%%%%%%%%%%%%%%%%%%%%
%%%%%%%%%%%%%%%%%%%%%%%%%%%%%%%%%%%%%%%%%%%%%%%%

\section{\label{sec:notation}Notation and preliminaries}
We fix our notations and review briefly some aspects of Lorentzian manifolds and differential operators on vector bundles.
We mainly follow \cite{Bar:2007zz,Bar:2011iu} and refer to these works for more details and references
to other literature.

A {\textit{Lorentzian manifold}} is a smooth and oriented connected $D$-dimensional manifold $M$ equipped with a smooth Lorentzian metric
$g$ of signature $(-,+,\dots,+)$. The associated volume form will be denoted by $\vol$. 
A time-oriented Lorentzian manifold will be called a {\textit{spacetime}}.
For every subset $A\subseteq M$ of a spacetime $M$ we denote the causal future/past of
$A$ by $J^{\pm}(A)$. 
A closed subset $A\subseteq M$ is called {\textit{spacelike compact}} if there exists a compact 
$C\subseteq M$ such that $A\subseteq J(C) :=J^+(C)\cup J^-(C)$. 
A {\textit{Cauchy surface}} in a spacetime $M$ is a subset $\Sigma\subseteq M$ 
which is met exactly once by every inextensible causal curve and a spacetime is called {\textit{globally hyperbolic}}
if and only if it contains a Cauchy surface.
We shall need the following theorem proven by Bernal and S{\'a}nchez \cite{Bernal:2004gm,Bernal:2005qf}:
\begin{theo}\label{theo:bernalsanchez}
Let $(M,g)$ be a globally hyperbolic spacetime.
\begin{itemize}
\item[(i)] Then there exists a smooth manifold $\Sigma$, a smooth one-parameter family of Riemannian metrics $\{g_t\}_{t\in\bbR}$
on $\Sigma$ and a smooth positive function $\vartheta$ on $\bbR\times \Sigma$, such that
$(M,g)$ is isometric to $(\bbR\times \Sigma,-\vartheta^2dt^2 \oplus g_t)$. Under this
 isometry each $\{t\}\times\Sigma$ corresponds to a
smooth spacelike Cauchy surface in $(M,g)$.
\item[(ii)] Let also $\widetilde{\Sigma}$ be a smooth spacelike Cauchy surface in $(M,g)$. Then 
there exists a smooth splitting $(M,g)\simeq (\bbR\times \Sigma,-\vartheta^2dt^2\oplus g_t)$ as in (i)
such that $\widetilde{\Sigma}$ corresponds to $\{0\}\times\Sigma$.
\end{itemize}
\end{theo} 
\sk

Let $V,W$ be a $\bbK$-vector bundles over $M$ with $\bbK=\bbR$ or $\bbC$.
A {\textit{differential operator}} of order $k$ is a linear map
$P:\Gamma^\infty(V)\to\Gamma^\infty(W)$, with $\Gamma^\infty(V),\Gamma^\infty(W)$ denoting the
$C^\infty(M)$-modules of sections of $V,W$, which in local coordinates $(x^0,\dots, x^{D-1})$ and a
local trivialisation of $V$ and $W$ looks like
\begin{flalign}
P=\sum\limits_{\vert\alpha \vert\leq k} A^\alpha(x) \frac{\partial^{\vert\alpha\vert}}{\partial x^{\alpha}}~.
\end{flalign}
Here $\alpha = (\alpha_0,\dots, \alpha_{D-1} )\in \bbN_0^{D}$ denotes a multi-index, $\vert\alpha\vert = \alpha_0+\dots+
\alpha_{D-1}$ is its length and $\frac{\partial^{\vert\alpha\vert}}{\partial x^{\alpha}}=
\frac{\partial^{\vert\alpha\vert}}{\partial (x^0)^{\alpha_0}\cdots \partial (x^{D-1})^{\alpha_{D-1}}}$. The $A^\alpha$ are smooth
 functions with values in the linear homomorphisms from the typical fibre of $V$ to the one of $W$.
 The {\textit{principal symbol}} $\sigma_P$ of $P$ associates to each covector 
 $\xi\in T_x^\ast M$ a homomorphism $\sigma_P(\xi):V_x\to W_x$ between the fibre $V_x$ and $W_x$ over $x\in M$. Locally,
 \begin{flalign}
 \sigma_P(\xi) = \sum\limits_{\vert\alpha\vert =k} A^\alpha(x) \,\xi^\alpha~,
 \end{flalign}
 where $\xi^\alpha = \xi_0^{\alpha_0}\dots \xi_{D-1}^{\alpha_{D-1}}$ and $\xi = \xi_\mu \,dx^\mu$ 
 (sum over $\mu=0,\dots,D-1$ understood).
 In addition to $\Gamma^\infty(V)$ we introduce the notations $\Gamma^\infty_0(V)$ for the sections
 of compact support and $\Gamma^\infty_\sc(V)$ for the sections of spacelike compact support.
\sk

Let now $\bbK=\bbR$ and let $\ip{~}{~}_V^{}$ be a non-degenerate
bilinear form on $V$, that is a family of 
non-degenerate bilinear maps $\ip{~}{~}_{V_x}^{}:V_x\times V_x\to \bbR$ on the fibres $V_x$, for all $x\in M$, 
that depend smoothly on $x$.
 We define the bilinear map $\ip{~}{~}_{\Gamma(V)}^{}$, for all sections $f,h\in\Gamma^\infty(V)$ with compact overlapping support,
\begin{flalign}
\ip{f}{h}_{\Gamma(V)}^{}:= \int_M\vol\,\ip{f}{h}_V^{}~.
\end{flalign}
Let us also assume that $W$ comes with a non-degenerate bilinear form $\ip{~}{~}_W^{}$.
Then every differential operator $P:\Gamma^\infty(V)\to\Gamma^\infty(W)$ of order $k$ has a unique {\textit{formal adjoint}}, 
i.e.~a differential operator $P^\dagger:\Gamma^\infty(W)\to\Gamma^\infty(V)$ of order $k$, such that
\begin{flalign}
\ip{P^\dagger f}{ h}_{\Gamma(V)}^{} = \ip{f}{Ph}_{\Gamma(W)}^{}~,
\end{flalign} 
for all $f\in\Gamma^\infty(W)$ and $h\in\Gamma^\infty(V)$ with compact overlapping support. 
If $V=W$, $\ip{~}{~}_V^{} = \ip{~}{~}_W^{}$
 and $P^\dagger =P$ we say that
$P$ is {\textit{formally self-adjoint}} (with respect to $\ip{~}{~}_V^{}$).
\begin{defi}
Let $P:\Gamma^\infty(V)\to \Gamma^\infty(V)$ be a differential operator on a vector bundle $V$ over a Lorentzian manifold $M$.
A {\textit{retarded/advanced Green's operator}} for $P$ is a continuous linear map $G_{\pm}:\Gamma^\infty_0(V)\to\Gamma^\infty(V)$
satisfying
\begin{itemize}
\item[(i)] $P\circ G_{\pm} =\id$,
\item[(ii)] $G_{\pm}\circ P\big\vert_{\Gamma^\infty_0(V)} = \id$,
\item[(iii)] $\supp(G_{\pm}f)\subseteq J^\pm(\supp(f))$ for any $f\in \Gamma^\infty_0(V)$.
\end{itemize}
\end{defi}
\sk

\begin{defi}
Let $P:\Gamma^\infty(V)\to\Gamma^\infty(V)$ be a differential operator on a vector bundle $V$ over
a globally hyperbolic spacetime $M$ with
a non-degenerate bilinear form $\ip{~}{~}_V^{}$.
\begin{itemize}
\item[(i)] We say that $P$ is {\textit{Green-hyperbolic}} if $P$ and $P^\dagger$ have Green's operators\footnote{We are grateful to Ko Sanders for pointing out that the existence of Green's operators for $P^\dagger$ does in general not follow from the existence of Green's operators for $P$.}.
\item[(ii)] We say that $P$ is {\textit{Cauchy-hyperbolic}} if the
  Cauchy problems for $P$ and $P^\dagger$ are well-posed.
\end{itemize}
\end{defi}
\sk

\begin{rem}
The Green's operators of a Green-hyperbolic operator on a globally hyperbolic spacetime
 are necessarily unique, see Remark 3.7 in \cite{Bar:2011iu}.
 Cauchy-hyperbolic operators are also Green-hyperbolic, but there are Green-hyperbolic operators
 that are not Cauchy-hyperbolic, see Section 2.7 in \cite{Bar:2011iu}.
\end{rem}
\sk

\begin{ex}
Let $M$ be a globally hyperbolic spacetime and $V$ a vector bundle over $M$.
\begin{itemize}
\item[1.)] A second-order differential operator $P$ on $V$ is called a {\textit{normally hyperbolic operator}} (also wave operator)
if its principal symbol is given by the inverse metric $g^{-1}$ times the identity on the fibre, $\sigma_P(\xi) =g^{-1}(\xi,\xi)\,\id$.
In other words, a differential operator is normally hyperbolic if and only if in local coordinates $x^\mu$ 
and a local trivialisation of $V$
\begin{flalign}
P =g^{\mu\nu}(x)\,\partial_\mu\partial_\nu +A^\mu(x)\,\partial_\mu + B(x)~,
\end{flalign}
where $A^\mu$ and $B$ smooth functions valued in the endomorphisms of the typical fibre of $V$.
\item[2.)] A first-order differential operator $P$ on $V$ is called of {\textit{Dirac-type}} if 
$P^2 = P\circ P$ is a normally hyperbolic operator.
\end{itemize}
The formal adjoints of normally hyperbolic operators and operators of
Dirac-type are again normally hyperbolic and of Dirac-type
respectively, and these two classes of differential operators are Green-hyperbolic and even Cauchy-hyperbolic, see
\cite{Bar:2007zz,Bar:2011iu,Muehlhoff:2010ra}. 
\end{ex}
\sk

As a last prerequisite we require the following lemma and theorem on properties of Green's operators. See 
Lemma 3.3 and Theorem 3.5 in \cite{Bar:2011iu} for the proofs.
\begin{lem}\label{lem:greenadjoint}
Let $M$ be a globally hyperbolic spacetime and $V$ a vector bundle over $M$ equipped with 
a non-degenerate bilinear form $\ip{~}{~}_V^{}$.
 Denote by $G_\pm$ the retarded/advanced Green's operators for
a Green-hyperbolic operator $P$ on $V$. Then the retarded/advanced Green's operators $G_\pm^\dagger$
for $P^\dagger$ satisfy, for all $f,h\in \Gamma^\infty_0(V)$,
\begin{flalign}
\ip{G_\mp^\dagger f}{h}_{\Gamma(V)}^{} = \ip{f}{G_\pm h}_{\Gamma(V)}^{}~.
\end{flalign}
In particular, if $P^\dagger=P$ is formally self-adjoint then
$\ip{G_\mp f}{h}_{\Gamma(V)}^{} = \ip{f}{G_\pm h}_{\Gamma(V)}^{}$, for all $f,h\in \Gamma^\infty_0(V)$.
\end{lem}
\sk
\begin{theo}\label{theo:complex}
Let $M$ be a globally hyperbolic spacetime, $V$ a vector bundle over $M$ and
$P$ a Green-hyperbolic operator on $V$. For $G_\pm$ being the retarded/advanced Green's operators for
$P$ we define the linear map $G:= G_+ -G_- : \Gamma^\infty_0(V)\to\Gamma^\infty_\sc(V)$.
Then the following sequence of linear maps is a complex, which is exact everywhere:
\begin{flalign}
\{0\} \stackrel{~}{\longrightarrow} \Gamma^\infty_0(V)\stackrel{P}{\longrightarrow} \Gamma^\infty_0(V) 
\stackrel{G}{\longrightarrow} \Gamma^\infty_\sc(V) \stackrel{P}{\longrightarrow} \Gamma^\infty_\sc(V)~.
\end{flalign}
\end{theo}
\sk

%%%%%%%%%%%%%%%%%%%%%%%%%%%%%%%%%%%%%%%%%%%%%%%%
%%%%%%%%%%%%%%%%%%%%%%%%%%%%%%%%%%%%%%%%%%%%%%%%

\section{\label{sec:class}Classical gauge field theories}
In this section we provide a general setting to describe classical gauge field theories.
This requires, of course, more structures compared to classical field theories which are not subject to gauge
invariance, i.e.~classical matter field theories.
Throughout this article all field theories are assumed to be real and non-interacting, i.e.~the dynamics is governed by a linear
equation of motion operator. The non-trivial coupling 
is thus only to fixed classical background fields, such as the gravitational field or background gauge fields.

Before investigating classical gauge field theories we first provide a 
definition of a classical matter field theory following the spirit of \cite{Bar:2007zz,Bar:2011iu}
and give some examples. 
\begin{defi}\label{def:matterfield}
A (real) {\textit{classical matter field theory}} is given by a triple $\big(M,V,P\big)$, where 
\begin{itemize}
\item $M$ is a globally hyperbolic spacetime 
\item $V$ is a real vector bundle over $M$ equipped  with
a non-degenerate bilinear form $\ip{~}{~}_V^{}$
\item $P:\Gamma^\infty(V)\to\Gamma^\infty(V)$ is a formally self-adjoint Green-hyperbolic operator  
\end{itemize}
We say that a classical matter field theory is {\textit{bosonic}} if $\ip{~}{~}_{V}^{}$ is symmetric and {\textit{fermionic}} if
$\ip{~}{~}_V^{}$ is antisymmetric.
\end{defi}
\sk
\begin{ex}[Klein-Gordon field]\label{ex:KG1}
Let $M$ be a globally hyperbolic spacetime and $V:=M\times \bbR$ be the trivial real line bundle.
 We equip $V$ with the canonical non-degenerate symmetric bilinear form $\ip{~}{~}_V^{}$, which is induced 
 from the inner product on the typical fibre $\bbR$ given by, for all $v_1,v_2\in\bbR$, 
\begin{flalign}
\ip{v_1}{v_2}_\bbR^{} = v_1\,v_2~.
\end{flalign}
The $C^\infty(M)$-module of sections $\Gamma^\infty(V)$ is isomorphic to $C^\infty(M)$.

Using the differential $\dd: \Omega^n(M)\to\Omega^{n+1}(M)$ and its formal adjoint 
$\delta:\Omega^{n}(M)\to\Omega^{n-1}(M)$, given by $\delta = (-1)^{n D+D}\,\ast \dd  \ast$ with $D=\dim(M)$
and $\ast$ denoting the Hodge operator, we define the Klein-Gordon operator of mass $m\in [0,\infty)$
\begin{flalign}
P:C^\infty(M)\to C^\infty(M)~,~~f\mapsto Pf=\delta\dd f + m^2f~.
\end{flalign}
This operator is formally self-adjoint with respect to $\ip{~}{~}_{V}^{}$ and 
normally hyperbolic, thus in particular also Green-hyperbolic.

This shows that the Klein-Gordon field is a bosonic classical matter field theory according 
to Definition \ref{def:matterfield}.
\end{ex}
\begin{ex}[Majorana field]\label{ex:Maj1}
For our spinor conventions see Appendix \ref{app:conventions} and for a general discussion
of spinor fields we refer to \cite{Sanders:2008kp}.
Let $M$ be a globally hyperbolic spacetime of dimension $D\text{~mod~} 8= 2,3,4$ 
 equipped with a spin structure and let $DM$ be the Dirac bundle. The typical fibre of $DM$ is given by
  $ \bbC^{2^{\lfloor D/2\rfloor}}$. We can use the charge conjugation map $^c:DM\to DM$ to define the real 
subbundle $V:=DM_\bbR := \big\{e\in DM : e^c=e \big\}$, which we call the Majorana bundle.
We equip the typical fibre $ \bbR^{2^{\lfloor D/2\rfloor}}$ of $DM_\bbR$ with the non-degenerate antisymmetric 
bilinear map, for all $v_1,v_2\in\bbR^{2^{\lfloor D/2\rfloor}}$,
\begin{flalign}\label{eqn:majprod}
\ip{v_1}{v_2}_{\bbR^{2^{\lfloor D/2\rfloor}}}^{} = i\,v_1^\mathrm{T}\,C\,v_2~,
\end{flalign}
where $C$ denotes the charge conjugation matrix, $i$ the imaginary unit 
and $^{\mathrm{T}}$ the transposition operation.
This induces a non-degenerate antisymmetric bilinear form $\ip{~}{~}_V^{}$ on $ V= DM_\bbR$.

Let us denote by $TM$ the tangent and by $T^\ast M$ the cotangent  bundle on $M$.
Using the connection $\nabla:\Gamma^\infty(V)\to\Gamma^\infty(V\otimes T^\ast M)$,
 which is induced by the Levi-Civita connection,
and the $\gamma$-matrix section $\gamma\in\Gamma^\infty\big(TM\otimes \End(V)\big)$, which is covariantly constant,
we define the Dirac operator $\slashed{\nabla}: \Gamma^\infty(V)\to\Gamma^\infty(V)$
 by the contraction of $\gamma$ and $\nabla$. In local coordinates we have $\slashed{\nabla} = \gamma^\mu\,\nabla_\mu$.
We further define the Dirac operator of mass $m\in [0,\infty)$ by
\begin{flalign}
P:\Gamma^\infty(V)\to\Gamma^\infty(V)~,~~f\mapsto Pf=\slashed{\nabla} f + m\,f ~.
\end{flalign}
 The operator $P$ is formally self-adjoint with respect to $\ip{~}{~}_V$ and of Dirac-type, thus in particular Green-hyperbolic.

This shows that the Majorana field is a fermionic classical field theory according to Definition
\ref{def:matterfield}.
\end{ex}
\sk

For a classical gauge field theory Definition \ref{def:matterfield} is not suitable, since firstly it does not encode the
notion of gauge invariance and secondly, as well-known, gauge invariance implies that the dynamics of gauge fields
can not be governed by hyperbolic operators.
To include  the missing structures we propose the following axioms:
\begin{defi}\label{def:gaugefield}
A {\textit{classical gauge field theory}} is given by a six-tuple $\big(M,V,W,P,K,T\big)$, where
\begin{itemize}
\item $M$ is a globally hyperbolic spacetime
\item $V$ and $W$ are real vector bundles over $M$  equipped  with
 non-degenerate bilinear forms $\ip{~}{~}_V^{}$ and $\ip{~}{~}_W^{}$
\item $P:\Gamma^\infty(V)\to\Gamma^\infty(V)$ is a formally self-adjoint differential operator
\item $K:\Gamma^\infty(W)\to\Gamma^\infty(V)$ is a differential operator satisfying
$P\circ K =0$ and $R:=K^\dagger\circ K$ Cauchy-hyperbolic for non-trivial $K\neq 0$
\item $T: \Gamma^\infty(W)\to \Gamma^\infty(V) $ is a differential operator, such that 
$\widetilde{P}:=P+T\circ K^\dagger$ is Green-hyperbolic and $Q:= K^\dagger\circ T$ is
 Green-hyperbolic for non-trivial $K\neq 0$
\end{itemize}
We say that a classical gauge field theory is {\textit{bosonic}} if $\ip{~}{~}_V$ is symmetric and {\textit{fermionic}} if
$\ip{~}{~}_V$ is antisymmetric.
\end{defi}
\sk
\begin{rem}
As the following examples will show, the objects appearing in the six-tuple $\big(M,V,W,P,K,T\big)$
describing a classical gauge field theory have the following physical interpretation:

Sections of the vector bundle $V$ describe configurations of the gauge field. The operator $P$ governs its dynamics and the formal self-adjointness of $P$ can be interpreted as saying that $P \psi=0$ are the Euler-Lagrange equations obtained from a quadratic action functional for $\psi$.
The operator $K$ generates gauge transformations by, for all $\psi \in\Gamma^\infty(V)$ and $\epsilon\in\Gamma^\infty(W)$,
 $\psi \mapsto \psi^\prime = \psi + K\epsilon$. Thus, sections of $W$ describe configurations of the gauge transformation
 parameters. The condition $P\circ K =0$ encodes the gauge invariance of the dynamics, in particular
 it implies that pure gauge configurations $K\epsilon\in\Gamma^\infty(V)$ solve the equation of motion.
 The condition $R:= K^\dagger\circ K$ Cauchy-hyperbolic is used to prove that $K^\dagger \psi =0$ is a consistent gauge
 fixing condition, i.e.~that any solution of $P\psi=0$ with spacelike compact support is gauge equivalent to 
 a solution in the kernel of $K^\dagger$, see Theorem \ref{theo:gaugeproperties} (iv).
The Green-hyperbolic operator $\widetilde{P}:= P+T\circ K^\dagger$ is the equation of motion operator after the
 canonical gauge fixing $K^\dagger \psi =0$. The Green-hyperbolic operator $Q:= K^\dagger\circ T$ ensures 
 that the canonical gauge fixing is compatible with time evolution.
 
 Even though $K^\dagger$ has also the interpretation of a gauge fixing operator, we want to stress that we do not perform
  any explicit gauge fixing and work completely in terms of gauge invariant quantities when discussing algebras of observables. This follows in particular from Proposition \ref{propo:solspacepairing} which implies that the canonical (anti)commutation relations of the gauge field do not depend on $\widetilde P$, but only on $P$. A related observation is that two classical gauge field theories which differ only in the operator $T$ can be considered to be equivalent, see Proposition \ref{propo:gaugevectorspace}.

Since for a given five-tuple $\big(M,V,W,P,K\big)$ the choice of $T$ seems to be non-unique
 in general and since in the following examples $T$ is usually read off from the five-tuple $\big(M,V,W,P,K\big)$ 
 rather than being given as an independent datum, a natural question is whether and under which additional
  assumptions a differential operator $T$ satisfying the last point of Definition \ref{def:gaugefield} exists for 
  every five-tuple $\big(M,V,W,P,K\big)$ satisfying the first four points of Definition \ref{def:gaugefield}. 
  Unfortunately, a satisfactory answer to this question, which would allow us to treat linear gauge theories solely 
  in terms of five-tuples $\big(M,V,W,P,K\big)$, seems to be non-trivial and is beyond the scope of this work. For this reason 
  we have chosen to consider $T$ as an additional datum in our following general treatment of linear gauge theories.

\end{rem}
\sk

Before providing non-trivial examples of classical gauge field theories we show that any 
 classical matter field theory is also a classical gauge field theory with trivial gauge structure 
$K$.
\begin{propo}
Let $\big(M,V,P\big)$ be a classical matter field theory and let $T: \Gamma^\infty(V)\to \Gamma^\infty(V)$ be an arbitrary differential operator. Then
$\big( M,V,V,P,K=0,T\big)$ is a classical gauge field theory with trivial gauge structure $K=0$.
\end{propo}
\begin{proof}
Since $K=0$ we also have $K^\dagger=0$. All conditions of Definition \ref{def:gaugefield}  are easily verified.
\end{proof}
The standard examples of linearised bosonic and fermionic gauge field theories 
also fit into Definition \ref{def:gaugefield}.
\begin{ex}[Linearised Yang-Mills field]\label{ex:YM1}
The Yang-Mills field should only serve as an illustrative example. This is why we restrict ourselves
to the case of trivial gauge bundles in order to simplify the discussion.

Let $M$ be a globally hyperbolic spacetime and $\mathfrak{g}$ be a real semisimple Lie algebra.
Let $W$ be the trivial vector bundle $W:=M\times \mathfrak{g}$ and $V:=W\otimes T^\ast M$, with
$T^\ast M$ denoting the cotangent bundle. 
We equip $W$ with the non-degenerate symmetric bilinear form $\ip{~}{~}_W^{}$ induced from the Killing form
on the typical fibre $\mathfrak{g}$, for all $w_1,w_2\in\mathfrak{g}$,
\begin{flalign}
\ip{w_1}{w_2}_\mathfrak{g}^{} = \mathrm{Tr}\big(\mathrm{ad}_{w_1}\, \mathrm{ad}_{w_2}\big)
\end{flalign}
and $V$ with the non-degenerate symmetric bilinear form $\ip{~}{~}_V^{}$ given 
by the product of $\ip{~}{~}_W^{}$ and the inverse metric $g^{-1}$ on $M$.
The $C^\infty(M)$-module of sections $\Gamma^\infty(W)$ is isomorphic to the $C^\infty(M)$-module
 of $\mathfrak{g}$-valued functions 
$C^\infty(M,\mathfrak{g})$ and $\Gamma^\infty(V)$ is isomorphic to the $\mathfrak{g}$-valued
one-forms $\Omega^1(M,\mathfrak{g})$. 

A Yang-Mills field in this setting is a section
$A\in  \Omega^1(M,\mathfrak{g})$. 
The curvature of $A$ is given by
$F=\dd A + \frac{1}{2}\,[A,A]\in \Omega^2(M,\mathfrak{g})$.
We define the covariant differential $\dd^A : \Omega^n(M,\mathfrak{g})\to
 \Omega^{n+1}(M,\mathfrak{g})$ by 
$\dd^A\eta := \dd \eta + [A,\eta]$ and denote its formal adjoint by $\delta^A: \Omega^n(M,\mathfrak{g})\to
 \Omega^{n-1}(M,\mathfrak{g})$. Explicitly, $\delta^A\eta = (-1)^{n\,D+D}\,\ast\dd^A\ast\eta$, where
 $\ast$ denotes the Hodge operator and $D=\dim(M)$.
  The Yang-Mills equation reads $\delta^A F=0$.

Let us now linearise the Yang-Mills field $A$ around a solution $A_0\in\Omega^1(M,\mathfrak{g})$ of the
Yang-Mills equation, i.e.~we write $A=A_0+\alpha$ with $\alpha\in \Omega^1(M,\mathfrak{g})$
and consider only terms linear in $\alpha$. The linearised curvature reads $F_\mathrm{lin}= F_0 +\dd^{A_0}\alpha$, where
$F_0$ is the curvature of $A_0$ and $\dd^{A_0}$ the covariant differential given by $A_0$.
The linearisation of the Yang-Mills equation  yields
\begin{flalign}
0=\delta^{A_0} F_0 + \delta^{A_0}\dd^{A_0} \alpha + (-1)^{D}\,\ast[\alpha,\ast F_0] 
= \delta^{A_0}\dd^{A_0} \alpha - \ast[\ast F_0,\alpha]~,
\end{flalign}
since $A_0$ is on-shell. We define the differential operator $P$ on
 $\Omega^1(M,\mathfrak{g})\simeq \Gamma^{\infty}(V)$, 
\begin{flalign}
P: \Omega^1(M,\mathfrak{g}) \to  \Omega^1(M,\mathfrak{g})~,~~
\alpha\mapsto P \alpha = \delta^{A_0}\dd^{A_0} \alpha - \ast[\ast F_0,\alpha]~.
\end{flalign}
It is formally self-adjoint with respect to $\ip{~}{~}_V^{}$.

The gauge invariance of the full (not linearised) theory is given by
 transformations $A\mapsto  A + \dd^A\epsilon$ labelled by $\epsilon\in C^\infty(M,\mathfrak{g})$.
 Notice that $C^\infty(M,\mathfrak{g})\simeq \Gamma^\infty(W)$. If we linearise the gauge transformations
 we obtain for all $\epsilon\in C^\infty(M,\mathfrak{g}) $ the transformation law $\alpha\mapsto 
\alpha + \dd^{A_0}\epsilon$. Let us define the  operator $K$ by
 \begin{flalign}
 K:C^\infty(M,\mathfrak{g})\to \Omega^{1}(M,\mathfrak{g})~,~~\epsilon\mapsto K\epsilon = \dd^{A_0}\epsilon~.
 \end{flalign} 
 It is a standard calculation to check that $P\circ K =0$, provided the background Yang-Mills field $A_0$ is
 on-shell, i.e.~$\delta^{A_0}F_0=0$.
 
We define further the operator 
\begin{flalign}
T: C^\infty(M,\mathfrak{g})\to \Omega^1(M,\mathfrak{g})~,~~\eta \mapsto T\eta = \dd^{A_0}\eta~.
\end{flalign}
Notice that $T=K$ and that $\widetilde{P} := P + T\circ K^\dagger =  \delta^{A_0}\circ \dd^{A_0} + \dd^{A_0}\circ\delta^{A_0} - \ast[\ast F_0,\,\cdot\,]$ is normally hyperbolic and thus in particular Green-hyperbolic.
We further obtain  $Q:= K^\dagger\circ T = \delta^{A_0}\circ \dd^{A_0}$, which is a normally hyperbolic
operator on $C^\infty(M,\mathfrak{g})$ and thus in particular Green-hyperbolic.
The operator $R:=K^\dagger\circ K=\delta^{A_0}\circ\dd^{A_0}$ agrees with $Q$ and is Cauchy-hyperbolic.
 
 This shows that the linearised Yang-Mills field on a trivial $\mathfrak{g}$-bundle is a bosonic classical gauge field theory according
 to Definition \ref{def:gaugefield}. 
\end{ex}
\begin{ex}[Linearised general relativity]\label{ex:gravity1}
The case of linearised $D{=}4$ general relativity in presence of a cosmological constant $\Lambda$ has been recently
studied in detail by Fewster and Hunt \cite{Fewster:2012bj}. We briefly show that this theory is a bosonic classical gauge
field theory according to Definition \ref{def:gaugefield} and refer to  \cite{Fewster:2012bj} for more details.
As in this paper we restrict ourselves to $D{=}4$ and employ a tensor index notation to simplify readability.

Let $M$ be a globally hyperbolic spacetime of dimension $D{=}4$. Let further $W:=T^\ast M$ be the cotangent bundle
and $V := \bigvee^2 T^\ast M $ be the bundle of symmetric contravariant tensors of rank $2$. The metric 
$g_{\mu\nu}\in \Gamma^\infty(V)$ of the globally hyperbolic spacetime
 $M$ is assumed to be a solution of the vacuum Einstein equations
$R_{\mu\nu} = \Lambda\,g_{\mu\nu}$, with $R_{\mu\nu}$ denoting the Ricci tensor of $g_{\mu\nu}$. 
We equip $W$ with the canonical 
non-degenerate symmetric bilinear form $\ip{~}{~}_W^{}$ induced by the inverse metric $g^{\mu\nu}$ 
on $M$ and $V$ with the non-degenerate symmetric bilinear form
\begin{flalign}
\ip{f}{h}_V = \overline{f}^{\mu\nu}\,h_{\mu\nu} = g^{\mu\rho}g^{\nu\sigma}\big(f_{\mu\nu}-\frac{1}{2}g_{\mu\nu}\,f\big)\,h_{\rho\sigma} 
= f^{\mu\nu}h_{\mu\nu} - \frac{1}{2}\,f\,h~,
\end{flalign}
where $f=f^\mu_\mu = g^{\mu\nu}f_{\mu\nu}$ is the trace and $\overline{\,\cdot\,}$ is called the trace-reversal operation.

Let us consider fluctuations $g_{\mu\nu}+\epsilon_{\mu\nu}$, with $\epsilon_{\mu\nu}
\in\Gamma^\infty(V)$, of the background metric. 
The equation of motion operator obtained by linearising the vacuum Einstein equations reads for the 
trace-reversed metric fluctuations $h_{\mu\nu} := \overline{\epsilon}_{\mu\nu} = \epsilon_{\mu\nu} -
\frac{1}{2}g_{\mu\nu} \,\epsilon$  
\begin{flalign}
P:\Gamma^\infty(V)\to\Gamma^\infty(V)~,~~h_{\mu\nu}\mapsto (Ph)_{\mu\nu} = g_{\mu\nu}\nabla^\rho\nabla^\sigma h_{\rho\sigma}+\square h_{\mu\nu} +2 \Lambda\,h_{\mu\nu} -2 \nabla^\rho\nabla_{(\mu}h_{\nu)\rho}~,
\end{flalign}
where $\nabla$ denotes the Levi-Civita connection corresponding to $g_{\mu\nu}$
 and $\square = \nabla^\mu\nabla_\mu= g^{\mu\nu}\nabla_\mu\nabla_\nu$ the 
d'Alembert operator. The parenthesis $(~)$ denotes symmetrisation of weight one.
It can be checked that $P$ is formally self-adjoint with respect to  $\ip{~}{~}_V$.

The gauge invariance of linearised general relativity is governed by the operator 
\begin{flalign}
K:\Gamma^\infty(W)\to \Gamma^{\infty}(V)~,~~w_\mu\mapsto (Kw)_{\mu\nu} =  \overline{\nabla_{(\mu} w_{\nu)}} = \nabla_{(\mu}w_{\nu)} -\frac{1}{2} g_{\mu\nu}\nabla^\rho w_\rho~.
\end{flalign}
The property $P\circ K =0$, which holds for backgrounds satisfying the on-shell condition $R_{\mu\nu}=\Lambda\,g_{\mu\nu}$, 
has already been verified in \cite{Fewster:2012bj}, see also \cite{Stewart:1974uz}. 
More precisely, the operators $P_{\text{FH}}$ and $K_{\text{FH}}$ of Fewster and Hunt are related
to ours by $P=-2\,P_{\text{FH}}\circ \overline{\,\cdot\,}$ and $K =\frac{1}{2}\, \overline{\,\cdot\,}\circ K_{\text{FH}}$
and from $P_{\text{FH}}\circ K_{\text{FH}} =0$ it follows $P\circ K = -P_{\text{FH}}\circ \overline{\,\cdot\,} \circ\overline{\,\cdot\,}\circ K_{\text{FH}} =  -P_{\text{FH}}\circ K_{\text{FH}} =0$, since the trace-reversal squares to the identity.
The formal adjoint of $K$ is given by, for all $h_{\mu\nu}\in \Gamma^\infty(V)$, $(K^\dagger h)_\mu =- \nabla^\nu h_{\mu\nu}$.

Let us further define the operator
\begin{flalign}
T:\Gamma^\infty(W)\to\Gamma^\infty(V)~,~~w_\mu\mapsto (Tw)_{\mu\nu} =-2 (Kw)_{\mu\nu} =- 2\left(\nabla_{(\mu} w_{\nu)}-\frac{1}{2}g_{\mu\nu} \nabla^\rho w_\rho  \right) ~.
\end{flalign}
For $\widetilde{P}:=P+ T\circ K^\dagger$  we obtain
\begin{flalign}
\widetilde{P}: \Gamma^\infty(V)\to\Gamma^\infty(V) ~,~~h_{\mu\nu}\mapsto (\widetilde{P}h)_{\mu\nu} =  \square h_{\mu\nu} -2\,R^{\rho~~\sigma}_{~\mu\nu~}h_{\rho\sigma}~,
\end{flalign}
where $R^{\rho~~\sigma}_{~\mu\nu~}$ is the Riemann tensor. This is a normally hyperbolic operator and thus in particular
Green-hyperbolic.
For $Q:= K^\dagger\circ T$ we obtain
\begin{flalign}
Q:\Gamma^\infty(W)\to\Gamma^\infty(W)~,~w_\mu \mapsto (Qw)_\mu = \square w_\mu +\Lambda w_\mu~,
\end{flalign}
which is also a normally hyperbolic operator and thus in particular Green-hyperbolic. 
The operator $R:=K^\dagger\circ K= -\frac{1}{2} Q$ is a multiple of a normally hyperbolic operator and in particular
 Cauchy-hyperbolic.

This shows that linearised general relativity in presence of a cosmological constant is a bosonic classical gauge field
theory according to Definition \ref{def:gaugefield}.
\end{ex}
\begin{ex}[Toy model: Fermionic gauge field]\label{ex:toy1}
Before introducing the Rarita-Schwinger gauge field as an example of a fermionic gauge field theory
in Example \ref{ex:RS1} we first discuss a simple toy model.

Let $M$ be a globally hyperbolic spacetime and let $\big(\bbR^{2m},\Omega\big)$, with $m\in\bbN$, be 
the symplectic vector space of dimension $2m$, 
i.e.~$\Omega$ is a non-degenerate antisymmetric $2m\times 2m$-matrix.
We define $W:=M\times \bbR^{2m}$ to be the trivial vector bundle 
and equip it with the non-degenerate antisymmetric bilinear form $\ip{~}{~}_W^{}$ induced from the
 symplectic structure on the typical fibre, for all $w_1,w_2\in \bbR^{2m}$,
\begin{flalign}
\ip{w_1}{w_2}_\Omega^{} := w_1^{\mathrm{T}}\Omega w_2~.
\end{flalign}
We further define $V:= W\otimes T^\ast M$, where $T^\ast M$ is the cotangent bundle, and equip it with the 
non-degenerate antisymmetric bilinear form $\ip{~}{~}_V^{}$ given by the product of $\ip{~}{~}_W^{}$ and
the inverse metric $g^{-1}$ on $M$.
The $C^\infty(M)$-module of sections $\Gamma^\infty(W)$ is isomorphic to the $C^\infty(M)$-module
$C^\infty(M,\bbR^{2m})$ and $\Gamma^\infty(V)$ is isomorphic to the $\bbR^{2m}$-valued one-forms 
$\Omega^1(M,\bbR^{2m})$.

We define the operator
\begin{flalign}
P:\Omega^1(M,\bbR^{2m})\to \Omega^1(M,\bbR^{2m})~,~~\alpha\mapsto P\alpha = \delta\dd\alpha~,
\end{flalign} 
which is formally self-adjoint with respect to $\ip{~}{~}_V^{}$. We further define
\begin{flalign}
K:C^\infty(M,\bbR^{2m})\to \Omega^1(M,\bbR^{2m})~,~~\epsilon \mapsto K\epsilon =\dd \epsilon~.
\end{flalign}
It obviously holds $P\circ K =0$ and the formal adjoint of $K$ is $K^\dagger =\delta$.
Defining the operator
\begin{flalign}
T:C^\infty(M,\bbR^{2m})\to \Omega^1(M,\bbR^{2m})~,~~\epsilon \mapsto T\epsilon =\dd \epsilon~,
\end{flalign}
we obtain that the operators $\widetilde{P}:= P + T\circ K^\dagger = \delta\circ\dd + \dd \circ \delta$ 
(on $\Omega^1(M,\bbR^{2m})$) and 
$Q:= K^\dagger\circ T = \delta\circ \dd$  (on $C^\infty(M,\bbR^{2m})$) are 
normally hyperbolic and thus in particular Green-hyperbolic.
Since $T=K$ we also have that $R:=K^\dagger\circ K=\delta\circ\dd$ is a normally hyperbolic operator on $C^\infty(M,\bbR^{2m})$
and in particular Cauchy-hyperbolic.

The six-tuple $\big(M,V,W,P,K,T\big)$ is thus a fermionic classical gauge field theory according to Definition 
\ref{def:gaugefield}.
\end{ex}
\begin{ex}[Rarita-Schwinger gauge field]\label{ex:RS1}
 Our model for the Rarita-Schwinger gauge field is inspired by $D{=}4$ simple supergravity, 
which we will briefly sketch. For details on supergravity we refer to \cite{VanNieuwenhuizen:1981ae,Nilles:1983ge,Wess:1992cp}.
The field content of this theory is the gravitational field, described by a vierbein $E$, and
the gravitino field $\Psi$. The action functional is given by a locally supersymmetric extension
of the Einstein-Hilbert action of general relativity.
Solutions of the corresponding equations of motion in a trivial gravitino background $\Psi=0$
are given by Ricci-flat Lorentzian manifolds $(M,g)$. We are interested in modelling 
linearised fluctuations of the gravitino field around these backgrounds.

As we have already seen in the Examples \ref{ex:YM1} and \ref{ex:gravity1}, the on-shell conditions
for the background fields are necessary to maintain gauge invariance of the linearised gauge field theory.
Thus, we are forced to assume that $M$ is a globally hyperbolic spacetime which is Ricci-flat and
 equipped with a spin structure. We take $D\text{~mod~}8 =2,3,4$ in order to have a suitable Majorana condition
 available, see Appendix \ref{app:conventions} for our spinor
 conventions.
The Rarita-Schwinger gauge field on more general spacetimes requires the coupling 
of supergravity to matter fields and will be discussed elsewhere.
We also assume that $D\geq 3$ to have a non-trivial equation of motion for 
the gravitino (otherwise the $\gamma^{\mu\nu\rho}$ defined below is
trivial; note that this is well in accord with the fact that gravity
in $D{=}2$ is not dynamical).
We define $W:= DM_\bbR$ to be the Majorana bundle (see Example \ref{ex:Maj1})
 and $V:= DM_\bbR\otimes T^\ast M$, where $T^\ast M$ denotes the  cotangent bundle.
 We equip $W$ with the canonical non-degenerate antisymmetric bilinear form $\ip{~}{~}_W^{}$, see
 (\ref{eqn:majprod}) for an expression on the typical fibre. It is convenient not to use  the
 supergravity gravitino $\Psi\in \Gamma^\infty(V)$ (linearised around the trivial configuration)
as the dynamical degrees of freedom, but to do a field redefinition instead. This is similar to the trace-reversal
 we have used in Example \ref{ex:gravity1}. Using the $\gamma$-section
  $\gamma\in\Gamma^\infty\big(TM\otimes\End(DM_\bbR)\big)$ we define the
   linear map $\widetilde{\,\cdot\,} :\Gamma^\infty(V)\to\Gamma^\infty(V)$,
  which is given in local coordinates by, for all $\psi\in\Gamma^\infty(V)$,
  $\widetilde{\psi}_\mu  := \psi_\mu -\frac{1}{D-2} \gamma_\mu \,\gamma^\nu\psi_\nu$,
  where $\gamma_\mu =  g_{\mu\rho}\,\gamma^\rho$. Notice that $\gamma^\mu\widetilde{\psi}_\mu 
  = -\frac{2}{D-2}\gamma^\mu\psi_\mu$
     and that $\widetilde{\,\cdot\,}$ is invertible via $\widetilde{\,\cdot\,}^{-1}$ given locally  by
     $\widetilde{\psi}^{-1}_\mu = \psi_\mu-\frac{1}{2}\gamma_\mu\gamma^\nu\psi_\nu$.
     We define the Rarita-Schwinger gauge field $\psi\in\Gamma^\infty(V)$ by the equation $\Psi = \widetilde{\psi}$, 
    where $\Psi\in\Gamma^\infty(V)$ is the linearised supergravity gravitino field.
 We equip $V$ with the non-degenerate bilinear form $\ip{~}{~}_V^{}$, which reads in local coordinates
 \begin{flalign}
 \ip{\psi_1}{\psi_2}_{V}^{} := \ip{\widetilde{\psi_1}_\mu}{\psi_2^\mu}_W^{} = 
 \ip{{\psi_1}_\mu}{\psi_2^\mu}_W^{} +\frac{1}{D-2} \ip{\gamma^\mu{\psi_1}_\mu}{\gamma^\nu{\psi_2}_\nu}_W^{}~.
 \end{flalign}
 Notice that $\ip{~}{~}_V^{}$ is antisymmetric. 
 
 The equation of motion for the linearised supergravity gravitino field $\Psi\in\Gamma^\infty(V)$ is 
  obtained by the supergravity action and it is given by the massless Rarita-Schwinger equation, which reads in local coordinates
 $\gamma^{\mu\nu\rho}\nabla_\nu\Psi_\rho =0$, where $\gamma^{\mu\nu\rho} = \gamma^{[\mu}\gamma^\nu\gamma^{\rho]}$,
 the parenthesis $[~]$ denotes antisymmetrisation of weight one and $\nabla$ is the connection on $V=DM_\bbR\otimes T^\ast M$
 induced by the Levi-Civita connection. For the redefined
 degrees of freedom $\psi\in\Gamma^\infty(V)$ with $\Psi =\widetilde{\psi}$ the dynamics
 is governed by the equation of motion operator $P$, given in local coordinates by
 \begin{flalign}\label{eqn:RSEOM}
 P:\Gamma^\infty(V)\to\Gamma^\infty(V)~,~~\psi_\mu \mapsto (P\psi)_\mu =
  \slashed{\nabla}\psi_\mu-\gamma_\mu\nabla^\nu\psi_\nu~. 
 \end{flalign} 
This operator is formally self-adjoint with respect to $\ip{~}{~}_V^{}$.
 
 The linearised local supersymmetry transformations act on the supergravity gravitino field $\Psi\in\Gamma^\infty(V)$
 by $\Psi_\mu \mapsto \Psi_\mu +\nabla_\mu\epsilon $, where
  $\epsilon\in\Gamma^\infty(W)$. For the redefined
  degrees of freedom $\psi\in\Gamma^\infty(V)$ with $\Psi = \widetilde{\psi}$ we
 obtain the operator $K$, given in local coordinates by
 \begin{flalign}
 K:\Gamma^\infty(W)\to\Gamma^\infty(V)~,~~\epsilon\mapsto (K\epsilon)_\mu = 
  \widetilde{\nabla_\mu\epsilon}^{-1}=
  \nabla_\mu\epsilon -\frac{1}{2}\gamma_\mu \slashed{\nabla}\epsilon~.
 \end{flalign}
 By a standard calculation one checks that $P\circ K =0$ if and only if the metric $g$ is Ricci-flat, which was exactly
 the on-shell condition imposed by supergravity. The formal adjoint of $K$ is given by, for all 
 $f\in\Gamma^\infty(V)$, $K^\dagger f = -\nabla^\mu f_\mu$.
 
 Let us further define the operator
 \begin{flalign}
 T:\Gamma^\infty(W)\to\Gamma^\infty(V)~,~~f\mapsto (Tf)_\mu = -\gamma_\mu f~.
 \end{flalign}
 Then $\widetilde{P}:= P+T\circ K^\dagger$ is simply the (twisted) Dirac operator on $V$, given in local coordinates
 by $(\widetilde{P} \psi)_\mu = \slashed{\nabla}\psi_\mu$. We further find that the operator 
 $Q:= K^\dagger\circ T$ is the Dirac operator on $W$ (remember that the section $\gamma$ is covariantly constant).
 These operators are of Dirac-type and thus in particular Green-hyperbolic. 
 For the operator $R:=K^\dagger\circ K$ we find, for all $\epsilon\in\Gamma^\infty(W)$,
 $R\epsilon= -\frac{1}{2}\nabla^\mu\nabla_\mu\epsilon$, where we have used that the metric $g$ is Ricci-flat.
 This is up to a constant prefactor a normally hyperbolic operator and thus in particular Cauchy-hyperbolic.
 
 This shows that the Rarita-Schwinger gauge field is a fermionic classical gauge field theory according
 to Definition \ref{def:gaugefield}. 
\end{ex}
\sk\sk

We collect important properties of classical gauge field theories which follow from
the Definition \ref{def:gaugefield} and will be required later for the construction and analysis of the algebra of observables.
Before, we have to introduce some notations:
\begin{defi}
Let $\big(M,V,W,P,K,T\big)$ be a classical gauge field theory. We define the following spaces:
\begin{itemize}
\item $\Ker_0(K^\dagger) := \big\{h\in\Gamma^\infty_0(V) : K^\dagger h=0\big\}$
\item $\Sol := \big\{f\in \Gamma^\infty_\sc(V): Pf=0\big\}$
\item $\mathcal{G} := K[\Gamma^\infty_\sc(W)] := \big\{ Kh : h\in\Gamma^\infty_\sc(W)\big\}$
\item $\widehat{\mathcal{G}}:= K[\Gamma^\infty(W)] \cap \Gamma^\infty_\sc(V)=\big\{Kh\in\Gamma^\infty_\sc(V):h\in\Gamma^\infty(W)\big\}$
\end{itemize}
Notice that $\mathcal{G}\subseteq \widehat{\mathcal{G}}\subseteq \Sol$, where the last inclusion is due to $P\circ K=0$.
We say that $\psi,\psi^\prime\in\Gamma^\infty_\sc(V)$ are {\textit{$\mathcal{G}$-gauge equivalent}}, if there exists a
$K \epsilon \in \mathcal{G}$ such that $\psi^\prime = \psi+K \epsilon$. 
Analogously, we say that $\psi,\psi^\prime\in\Gamma^\infty_\sc(V)$ are {\textit{$\widehat{\mathcal{G}}$-gauge equivalent}},
 if there exists a $K \epsilon \in \widehat{\mathcal{G}}$ such that $\psi^\prime = \psi+K \epsilon$. Since the inclusion
 $\mathcal{G}\subseteq \widehat{\mathcal{G}}$ holds true, $\mathcal{G}$-gauge equivalence 
 implies $\widehat{\mathcal{G}}$-gauge equivalence.
\end{defi}
\sk

\begin{theo}\label{theo:gaugeproperties}
Let $\big(M,V,W,P,K,T\big)$ be a classical gauge field theory with $\widetilde{P}:= P+T\circ K^\dagger$,
$Q:=K^\dagger\circ T$ and $R:=K^\dagger\circ K$. Let us denote by 
$G_\pm^{\widetilde{P}}: \Gamma^\infty_0(V)\to\Gamma^\infty(V)$ 
the retarded/advanced Green's operators for $\widetilde{P}$. 
In case of $K\neq 0$ we denote by $G_\pm^Q,G_\pm^{R}:\Gamma_0^\infty(W)\to\Gamma^\infty(W)$ the
 retarded/advanced Green's operators
for $Q$ and $R$, respectively. Then the following hold true:
\begin{itemize}
\item[(i)] $K^\dagger\circ \widetilde{P} = Q\circ K^\dagger$  and $\widetilde{P}\circ K = T\circ R$.
\item[(ii)] If $K\neq 0$, then $K^\dagger\circ G_\pm^{\widetilde{P}} = G_\pm^Q\circ K^\dagger$ on $\Gamma^\infty_0(V)$
and $K\circ G_{\pm}^{R} = G^{\widetilde{P}}_\pm\circ T$ on $\Gamma^\infty_0(W)$.
\item[(iii)] $G^{\widetilde{P}} := G^{\widetilde{P}}_+ - G^{\widetilde{P}}_-$ satisfies, for all $f,h\in \Ker_0(K^\dagger)$,
\begin{flalign}\label{eqn:skewker}
\ip{f}{G^{\widetilde{P}}h}_{\Gamma(V)}^{} = -\ip{G^{\widetilde{P}}f}{h}_{\Gamma(V)}^{}~.
\end{flalign}
That is, $G^{\widetilde{P}}$ is formally skew-adjoint with respect to $\ip{~}{~}_V^{}$ on the kernel $\Ker_0(K^\dagger)\subseteq \Gamma^\infty_0(V)$.
\item[(iv)] Any $\psi\in\Gamma^\infty_\sc(V)$ is $\mathcal{G}$-gauge equivalent to a $\psi^\prime \in\Gamma^\infty_\sc(V)$
satisfying $K^\dagger \psi^\prime =0$. 

 In particular, any $\psi\in\Sol$ is $\mathcal{G}$-gauge equivalent to a $\psi^\prime\in \Sol$ satisfying $K^\dagger \psi^\prime =0$
and thus also $\widetilde{P}\psi^\prime =0$.
\item[(v)] Any $\psi\in\Sol$ satisfying $K^\dagger \psi=0$ is $\mathcal{G}$-gauge equivalent
to $G^{\widetilde{P}}h$ for some $h\in\Ker_0(K^\dagger)$.
\item[(vi)] 
Let $h\in\Ker_0(K^\dagger)$, then $G^{\widetilde{P}}h\in\mathcal{G}$  
 if and only if $h\in P[\Gamma^\infty_0(V)]$.
\item[(vii)] 
Let $T^\prime :\Gamma^\infty(W)\to \Gamma^\infty(V)$ be an arbitrary differential operator 
such that replacing $T$ by $T^\prime$ is a classical gauge field theory and let
 $\widetilde{P}^\prime := P+T^\prime\circ K$. Then 
 $\ip{f}{G^{\widetilde{P}^\prime}_\pm h}_{\Gamma(V)^{}}=\ip{f}{G^{\widetilde{P}}_\pm h}_{\Gamma(V)}^{}$, 
 for all $f,h\in \Ker_0(K^\dagger)$.
\end{itemize}
\end{theo}
\begin{proof}
Proof of (i): Since $P$ is formally self-adjoint and $P\circ K =0$ we obtain $K^\dagger\circ P =0$. It follows
$K^\dagger \circ\widetilde{P} =  K^\dagger\circ T\circ K^\dagger = Q\circ K^\dagger$
and $\widetilde{P}\circ K = T\circ K^\dagger \circ K=T\circ R$.
\sk

Proof of (ii):  Using (i) we obtain,  for all $h\in\Gamma^\infty_0(W)$ and $f\in\Gamma^\infty_0(V)$,
\begin{flalign}
\nn \ip{h}{K^\dagger G_\pm^{\widetilde{P}}f}_{\Gamma(W)}^{} &= \ip{Q^\dagger G_\mp^{Q^\dagger}h}{K^\dagger G_\pm^{\widetilde{P}}f}_{\Gamma(W)}^{} 
= \ip{G_\mp^{Q^\dagger}h}{Q K^\dagger G_\pm^{\widetilde{P}}f}_{\Gamma(W)}^{} \\
&= \ip{G_\mp^{Q^\dagger}h}{K^\dagger\widetilde{P} G_\pm^{\widetilde{P}}f}_{\Gamma(W)}^{} 
=\ip{G_\mp^{Q^\dagger}h}{K^\dagger f}_{\Gamma(W)}^{} = \ip{h}{G_\pm^Q K^\dagger f}_{\Gamma(W)}^{}~,
\end{flalign}
where we also have used Lemma \ref{lem:greenadjoint} in the last equality.
The hypothesis now follows from the non-degeneracy of $\ip{~}{~}_W^{}$.
The other identity is proven analogously.
\sk

Proof of (iii): For $K=0$ we have $\widetilde{P}=P$ and the hypothesis follows from the fact that $P$ was assumed to
be formally self-adjoint and Lemma \ref{lem:greenadjoint}. Let us now assume that $K\neq 0$ and consider
 $f,h\in\Ker_0(K^\dagger)$. From (ii) we obtain 
$K^\dagger G^{\widetilde{P}}_\pm f = G^Q_\pm K^\dagger f=0$ and similarly $K^\dagger G^{\widetilde{P}}_\pm h=0$.
Thus,
\begin{flalign}
\nn \ip{f}{G^{\widetilde{P}}_\pm h}_{\Gamma(V)} &= \ip{\widetilde{P} G^{\widetilde{P}}_\mp f}{G^{\widetilde{P}}_\pm h}_{\Gamma(V)}
=\ip{P G^{\widetilde{P}}_\mp f}{G^{\widetilde{P}}_\pm h}_{\Gamma(V)} = \ip{G^{\widetilde{P}}_\mp f}{PG^{\widetilde{P}}_\pm h}_{\Gamma(V)}\\
&=\ip{G^{\widetilde{P}}_\mp f}{\widetilde{P}G^{\widetilde{P}}_\pm h}_{\Gamma(V)} = \ip{G^{\widetilde{P}}_\mp f}{ h}_{\Gamma(V)}~,
\end{flalign}
where we have used in the second and fourth equality that on $\Ker(K^\dagger)$ the operator $\widetilde{P}$ equals $P$
and in the third equality that $P$ is formally self-adjoint. This in particular shows (\ref{eqn:skewker}).
\sk

Proof of (iv): Let $\psi\in\Gamma^\infty_\sc(V)$ be arbitrary and let $\epsilon\in\Gamma_\sc^\infty(W)$.
We define $\psi^\prime := \psi + K\epsilon$ and obtain from the condition $K^\dagger \psi^\prime =0$ the equation
$K^\dagger K \epsilon = - K^\dagger \psi$. Since $K^\dagger \psi\in\Gamma^\infty_\sc(W)$ and $R=K^\dagger \circ K$ was assumed to
be Cauchy-hyperbolic this equation has a solution $\epsilon\in\Gamma^\infty_\sc (W)$, see \cite[Chapter 3, Corollary 5]{Bar:2009zzb}
for a discussion of how to treat inhomogeneities of non-compact support.
It then holds that $\psi^\prime = \psi+K\epsilon \in \Gamma_\sc^\infty(V)$ with $K^\dagger \psi^\prime =0$
and $K\epsilon\in\mathcal{G}$.
\sk

Proof of (v): We first note that as a consequence of (ii) and Theorem \ref{theo:complex}
we obtain that $G^{\widetilde{P}}h$ with $h\in\Gamma^\infty_0(V)$
satisfies $K^\dagger G^{\widetilde{P}}h =G^QK^\dagger h=0$ if and only if $K^\dagger h\in Q[\Gamma^\infty_0(W)]$.

Let now $\psi\in \Sol$ be such that $K^\dagger \psi =0$. As a consequence,
$\widetilde{P}\psi =0$ and since $\widetilde{P}$ is Green-hyperbolic there is a $h\in\Gamma^\infty_0(V)$ such that
$\psi=G^{\widetilde{P}}h$, see Theorem \ref{theo:complex}. Due to the argument above, 
we have $K^\dagger h = Q k$ for some $k\in \Gamma^\infty_0(W)$.
Let us consider the following $\mathcal{G}$-gauge transformation
\begin{flalign}
\psi-K G^{R} k \stackrel{\text{(ii)}}{=} \psi - G^{\widetilde{P}}Tk = G^{\widetilde{P}}\big(h-Tk\big)~.
\end{flalign}
Defining $h^\prime := h-Tk$ we have shown that $\psi$ is $\mathcal{G}$-gauge equivalent to $G^{\widetilde{P}}h^\prime$
with $K^\dagger h^\prime = K^\dagger h-K^\dagger Tk = Qk-Qk=0$, i.e.~$h^\prime\in \Ker_0(K^\dagger)$.
\sk

Proof of (vi): If $h = Pf \in P[\Gamma^\infty_0(V)]$ 
then $G^{\widetilde{P}}h = G^{\widetilde{P}}Pf = -G^{\widetilde{P}}TK^\dagger f = -K G^{R}K^\dagger f$
is an element in $\mathcal{G}$. To show the other direction, let $h\in\Ker_0(K^\dagger)$ be such that there is
a $k\in\Gamma^\infty_\sc(W)$ satisfying $G^{\widetilde{P}}h =Kk$. It follows that 
$K^\dagger Kk =0$ and since $R=K^\dagger\circ K$ is assumed to by Cauchy-hyperbolic
 there is by Theorem \ref{theo:complex} an $f\in\Gamma^\infty_0(W)$ such that $k=G^{R}f$.
 Using (ii) we obtain $K k = K G^{R} f = G^{\widetilde{P}}Tf = G^{\widetilde{P}}h$,
 which implies $h-Tf =\widetilde{P}q$ for some $q\in\Gamma^\infty_0(V)$.
 The condition $K^\dagger h=0$ further leads us to $-K^\dagger Tf =QK^\dagger q$, 
 i.e.~$Q\big(K^\dagger q + f\big)=0$, and since $f$ and $q$ are of compact support we have by Theorem \ref{theo:complex}
 $f=-K^\dagger q$. Thus, $h=Tf + \widetilde{P}q = -T K^\dagger q + \widetilde{P}q = Pq$.

Proof of (vii): For arbitrary $f,h\in\Ker_0(K^\dagger)$ we compute using (iii) and (ii)
\begin{flalign}
\nn \ip{f}{G_\pm^{\widetilde{P}}h}_{\Gamma(V)}^{}&=\ip{G_\mp^{\widetilde{P}}f}{h}_{\Gamma(V)}^{}=\ip{G_\mp^{\widetilde{P}}f}{\widetilde{P}^\prime G_\pm^{\widetilde{P}^\prime}h}_{\Gamma(V)}^{}\\
&=\ip{G_\mp^{\widetilde{P}}f}{\widetilde P G_\pm^{\widetilde{P}^\prime}h}_{\Gamma(V)}^{}+\ip{G_\mp^{\widetilde{P}}f}{\big(T^\prime-T\big)K^\dagger G_\pm^{\widetilde{P}^\prime}h}_{\Gamma(V)}^{}=\ip{f}{G_\pm^{\widetilde{P}^\prime}h}_{\Gamma(V)}^{}~.
 \end{flalign}
\end{proof}

%%%%%%%%%%%%%%%%%%%%%%%%%%%%%%%%%%%%%%%%%%%%%%%%
%%%%%%%%%%%%%%%%%%%%%%%%%%%%%%%%%%%%%%%%%%%%%%%%

\section{\label{sec:algebra}Gauge invariant on-shell algebra of observables}
The goal of this section is to construct from the data of a classical gauge field theory 
$\big(M,V,W,P,K,T\big)$ a suitable quantum algebra of gauge invariant observables describing the quantized gauge field theory.
We will first review the quantization of bosonic and fermionic matter field theories
and then extend these constructions to gauge field theories. We again follow the spirit of \cite{Bar:2011iu},
where also more details on bosonic and fermionic quantization can be found. 

The strategy to quantize a bosonic (fermionic) matter field theory $\big(M,V,P\big)$ is to first associate
to it a suitable symplectic (inner product) space, which is then quantized in terms of a CCR (CAR) representation.
\begin{propo}\label{propo:mattervectorspace}
Let $\big(M,V,P\big)$ be a classical matter field theory, denote the Green's operators 
for $P$ by $G_\pm$ and $G:=G_+-G_-$.
We define the vector space $\EE:= \Gamma_0^{\infty}(V)/P[\Gamma_0^\infty(V)]$
and the bilinear map
\begin{flalign}\label{eqn:sigmamat}
\tau: \EE\times \EE \to\bbR ~,~~([f],[h])\mapsto \tau\big([f],[h]\big) =\ip{f}{Gh}_{\Gamma(V)}^{} = \int_M\vol\,\ip{f}{Gh}^{}_V~.
\end{flalign} 
Then the map $\tau$ is well-defined and weakly non-degenerate. If further $\big(M,V,P\big)$ 
is bosonic, then $\tau$ is antisymmetric, i.e.~$(\EE,\tau)$ is a symplectic vector space.
If $\big(M,V,P\big)$ is fermionic, then $\tau$ is symmetric, 
i.e.~$(\EE,\tau)$ is an (i.g.~indefinite) inner product space.
\end{propo}
\begin{proof}
The map $\tau$ is well-defined, since $G$ is formally skew-adjoint with respect to
$\ip{~}{~}_V^{}$ (see Lemma \ref{lem:greenadjoint}) and we have $G\circ P =0$ on $\Gamma^\infty_0(V)$.

We now show that $\tau$ is weakly non-degenerate. Notice that because of the non-degeneracy of
$\ip{~}{~}_V^{}$ the condition $\ip{f}{Gh}_{\Gamma(V)}^{} =0$, for all $f\in\Gamma^\infty_0(V)$, implies
$Gh=0$. By Theorem \ref{theo:complex} there exists $k\in \Gamma^\infty_0(V)$, such that
$h=Pk$, meaning that $[h]=[0]$. Thus, $\tau$ is weakly non-degenerate.

Using again the skew-adjointness of $G$ and the symmetry (antisymmetry) of $\ip{~}{~}_V^{}$ for 
a bosonic (fermionic) matter field theory we obtain, for all $[f],[h]\in \EE$,
\begin{flalign}
\tau\big([f],[h]\big) = \ip{f}{Gh}_{\Gamma(V)}^{} = -\ip{Gf}{h}_{\Gamma(V)}^{} = \mp \ip{h}{Gf}_{\Gamma(V)}^{} 
=\mp \,\tau\big([h],[f]\big)~,
\end{flalign}
where $-$ is for bosonic and $+$ for fermionic theories. 
\end{proof}
For a physically and also mathematically consistent quantization of fermionic field theories
we have to demand further a positivity condition on $\tau$. See the Remarks \ref{rem:degenerate} and
 \ref{rem:positive} below for a detailed comment.
\begin{defi}\label{def:matterpositive}
A fermionic classical matter field theory $\big(M,V,P\big)$ is of {\textit{positive type}} if
$(\EE,\tau)$ is a (real) pre-Hilbert space, i.e.~the map $\tau$ is positive definite.
\end{defi}
\sk
We provide examples of fermionic classical matter field theories of positive type in the next section.
Any bosonic classical matter field theory can be quantized in terms of a CCR-representation.
\begin{defi}\label{def:CCR}
A {\textit{CCR-representation}} of a symplectic vector space $(\EE,\tau)$ is a pair
$(\w,A)$, where $A$ is a unital $C^\ast$-algebra and  $\w:\EE\to A$ is a map satisfying:
\begin{itemize}
\item[(i)] $A=C^\ast(\w(\EE))$,
\item[(ii)] $\w(0)=\1$,
\item[(iii)] $\w(f)^\ast = \w(-f)$,
\item[(iv)] $\w(f+h) = e^{i\,\tau(f,h)/2}~\w(f)\,\w(h)$,
\end{itemize}
for all $f,h\in \EE$.
\end{defi}
\sk

Furthermore, any fermionic classical matter field theory of positive type can be quantized in terms of a CAR-representation.
\begin{defi}\label{def:CAR}
A (self-dual) {\textit{CAR-representation}} of a  pre-Hilbert space $(\EE,\tau)$ over $\bbR$ is
a pair $(\b,A)$, where $A$ is a unital $C^\ast$-algebra and $\b:\EE\to A$ is a linear map satisfying:
\begin{itemize}
\item[(i)] $A=C^\ast(\b(\EE))$,
\item[(ii)] $\b(f)^\ast = \b(f) $,
\item[(iii)] $\big\lbrace \b(f),\b(h) \big \rbrace = \tau(f,h)\,\1$,
\end{itemize} 
for all $f,h\in \EE$.
\end{defi}
\sk

The following theorem is proven in \cite{Bar:2011iu,Bar:2007zz}.
\begin{theo}\label{theo:uniqueness}
There exists up to $C^\ast$-isomorphism a unique CCR-representation
  (CAR-representation) for every symplectic vector space (pre-Hilbert space). 
\end{theo}
\begin{rem}\label{rem:degenerate}
For defining the CCR-representation we have assumed that the map $\tau$ is weakly non-degenerate.
While for a bosonic classical matter field theory this is automatically given by Proposition \ref{propo:mattervectorspace},
this condition turns out to be too restrictive for gauge field theories, see Section \ref{sec:degeneracy}
for a discussion. The quantization of a pre-symplectic vector space $(\EE,\tau)$ can always be performed 
in terms of a field polynomial algebra. However, one looses the $C^\ast$-algebra property
 when making this choice. Fortunately, in \cite{degenerateCCR} the existence and uniqueness of
 the Weyl algebra for a generic pre-symplectic vector space has been proven.
 This means that Definition \ref{def:CCR} can be extended to any pre-symplectic vector space and
 the result of Theorem \ref{theo:uniqueness} is unaltered in this case. We refer to \cite{degenerateCCR} for 
 details on Weyl algebras of degenerate pre-symplectic vector spaces.
 We finish this remark by noting that a similar result for degenerate pre-Hilbert spaces and their
 CAR-quantization are not known to us.
\end{rem} 
\begin{rem}\label{rem:formalnotation}
This remark is quite standard, however, it is essential for understanding our construction of the
algebra of observables for a gauge field theory.

  Let $(\EE,\tau)$ be the symplectic vector space associated to a bosonic classical matter field theory $\big(M,V,P\big)$,
  i.e.~$\EE=\Gamma^\infty_0(V)/P[\Gamma^\infty_0(V)]$ and $\tau$ as given in (\ref{eqn:sigmamat}). The Weyl symbols
  $\w([f])$, $[f]\in\EE$, are physically interpreted as quantizations of the following functionals $\w_f$, $f\in\Gamma^\infty_0(V)$,
   on the configuration space  $\Gamma^\infty(V)$ of the classical matter field theory
  \begin{flalign}\label{eqn:formal1}
  \w_f : \Gamma^\infty(V)\to \bbC ~,~~\psi \mapsto \w_f[\psi] = e^{i\,\ip{\psi}{f}_{\Gamma(V)}^{}}~.
  \end{flalign}
  The on-shell condition $P\psi=0$ is then encoded on the level of the functionals by identifying
  $\w_{Pf}$ by the constant functional $\1$ (use (\ref{eqn:formal1}) and that $P$ is formally self-adjoint).
  The functionals on the on-shell configuration space are thus labelled by equivalence classes, i.e.~elements
   in $\EE=\Gamma^\infty_0(V)/P[\Gamma^\infty_0(V)]$.

An analogous interpretation holds for a fermionic matter field theory, where the symbols $\b([f])$, $[f]\in\EE$, are interpreted
as quantizations of the functionals
\begin{flalign}
\b_f:\Gamma^\infty(V)\to\bbR~,~~\psi\mapsto \b_f[\psi] = \ip{\psi}{f}_{\Gamma(V)}^{}~,
\end{flalign}
with $f\in\Gamma^\infty_0(V)$. The on-shell condition $P\psi=0$ is encoded here by identifying
the functionals $\b_{Pf}$, $f\in\Gamma_0^\infty(V)$, with $0$, giving rise to the
vector space $\EE = \Gamma^\infty_0(V)/P[\Gamma^\infty_0(V)]$ which labels the functionals on the
 on-shell configuration space.
\end{rem}
\begin{rem}\label{rem:positive}
We give a physical motivation for the positivity requirement for fermionic matter field theories
given in Definition \ref{def:matterpositive}. Take any $[f]\in\EE$ and consider the corresponding 
symbol $\b([f])$. From Definition \ref{def:CAR} (ii) and (iii) it follows that
\begin{flalign}\label{eqn:argument}
\big\{\b([f]),\b([f])\big\} = 2\,\b([f])^\ast\,\b([f]) = \tau([f],[f])\,\1~.
\end{flalign}
 Assume that we have a representation of this algebra on a Hilbert space 
 and let $\vert\Psi\rangle$ be any normalised vector $\langle \Psi\vert\Psi\rangle =1$.
 Taking the expectation value of both sides of (\ref{eqn:argument}) gives us the equality
 $\tau([f],[f]) = 2\,\langle \b([f])\Psi\vert \b([f])\Psi \rangle $. If now $\tau([f],[f])<0$
 the  vector $\vert \b([f])\Psi\rangle$ has a negative norm square, which is inconsistent with the Hilbert space
 assumption. In case $\tau([f],[f])=0$ the Hilbert space vector
  $\vert \b([f])\Psi\rangle$ has zero norm and since $\vert\Psi\rangle$ has been an arbitrary normalised
  vector the operator associated to $\b([f])$ is  the zero operator in any Hilbert space representation. 
\end{rem}
\sk

Let us now consider a classical gauge field theory $\big(M,V,W,P,K,T\big)$.
The goal is to construct a pre-symplectic vector space for bosonic
and a possibly indefinite inner product space for fermionic classical gauge field theories.
Following the interpretation of Remark \ref{rem:formalnotation} we 
are thus looking for a suitable vector space of smearing functions.
It turns out to be convenient to directly encode gauge invariance 
on the level of this vector space, leading later to a quantization of only the gauge invariant
degrees of freedom. Let us for example consider a bosonic classical
gauge field theory: We can again consider functionals on the off-shell configuration space as in 
(\ref{eqn:formal1}). Such a functional $\w_f$ is gauge invariant, i.e.~independent on whether we evaluate it
on $\psi$ or $\psi+K\epsilon$ with $\epsilon\in\Gamma^\infty(W)$, if and only if $K^\dagger f=0$.
Indeed,
\begin{flalign}
\w_f[\psi+K\epsilon] = e^{i\,\ip{\psi+K\epsilon}{f}_{\Gamma(V)}^{}} = e^{i\,\ip{\psi}{f}_{\Gamma(V)}^{} + i\,\ip{\epsilon}{K^\dagger f}_{\Gamma(W)}^{}}=\w_f[\psi]~,
\end{flalign}
for all $\epsilon\in\Gamma^\infty(W)$ if and only if $K^\dagger f=0$.
Thus, in order to capture the gauge invariant degrees of freedom we should
consider instead of $\Gamma^\infty_0(V)$ only the kernel $\Ker_0(K^\dagger) \subseteq \Gamma_0^\infty(V)$
 of $K^\dagger$ when formulating the space $\EE$ for gauge theories. The implementation of the on-shell condition
is then a suitable quotient by the equation of motion operator. This construction can be performed and
a well-defined pre-symplectic structure  (respectively, indefinite inner product structure) can
be defined on $\EE$ for bosonic (respectively fermionic) gauge field theories.
\begin{propo}\label{propo:gaugevectorspace}
Let $\big(M,V,W,P,K,T\big)$ be a classical gauge field theory with $\widetilde{P}:=P+T\circ K^\dagger$,
$Q:= K^\dagger\circ T$ and $R:= K^\dagger\circ K$. Let us denote by $G^{\widetilde{P}}_\pm$ the
 retarded/advanced Green's operators for $\widetilde{P}$ and  $G^{\widetilde{P}}:= G^{\widetilde{P}}_+ - G^{\widetilde{P}}_-$. 
We define the vector space $\EE:= \Ker_0(K^\dagger)/P[\Gamma^\infty_0(V)]$
and the bilinear map
\begin{flalign}\label{eqn:sigma}
\tau: \EE\times \EE \to\bbR ~,~~([f],[h])\mapsto \tau\big([f],[h]\big) =\ip{f}{G^{\widetilde{P}}h}_{\Gamma(V)}^{} 
= \int_M\vol\,\ip{f}{G^{\widetilde{P}}h}^{}_V~.
\end{flalign} 
Then the map $\tau$ is well-defined. Furthermore, $\tau$ is antisymmetric for bosonic gauge field theories and
symmetric for fermionic ones. Finally, let $T^\prime: \Gamma^\infty(W)\to \Gamma^\infty(V)$ be an arbitrary differential operator such that $\big(M,V,W,P,K,T^\prime \big)$ is a classical gauge field theory and let $\tau^\prime: \EE\times \EE \to\bbR$ be defined in analogy to \eqref{eqn:sigma} by means of $G^{\widetilde{P}^\prime}$ with $\widetilde{P}^\prime:=P+T^\prime\circ K^\dagger$. Then $\tau^\prime = \tau$.
\end{propo}
\begin{proof}
For the trivial case $K=0$ the proof is as in Proposition \ref{propo:mattervectorspace}. In particular, 
the vector space $\EE$ and the map $\tau$ are then exactly those of a classical matter field theory. 
So let us assume that $K\neq 0$.
According to Theorem \ref{theo:gaugeproperties} (iii) $G^{\widetilde{P}}$ is formally skew-adjoint with 
respect to $\ip{~}{~}_V^{}$ on $\Ker_0(K^\dagger)$. That $\tau$ is well-defined follows from this fact and the following calculation,
for all $f\in \Ker_0(K^\dagger)$ and $h\in\Gamma^\infty_0(V)$,
\begin{flalign}
\nn \ip{f}{G^{\widetilde{P}}Ph}_{\Gamma(V)}^{} &= \ip{f}{G^{\widetilde{P}}(\widetilde{P}-TK^\dagger) h}_{\Gamma(V)}^{}
= -\ip{f}{ K G^{R}K^\dagger h}_{\Gamma(V)}^{}\\
&=-\ip{K^\dagger f}{ G^{R}K^\dagger h}_{\Gamma(W)}^{}=0~,
\end{flalign}
where in the second equality we have used Theorem \ref{theo:gaugeproperties} (ii) and $G^{\widetilde{P}}\widetilde{P}h=0$.

The antisymmetry (symmetry) of $\tau$ for bosonic (fermionic) gauge field theories is proven as in the proof of 
Proposition \ref{propo:mattervectorspace}.

The last statement follows immediately from Theorem \ref{theo:gaugeproperties} (vii).
\end{proof}
In contrast to classical matter field theories we can not guarantee that the map $\tau$  is weakly non-degenerate for a classical 
gauge field theory. 
For bosonic gauge field theories this is mathematically not problematic, since the CCR-representation
of Definition \ref{def:CCR} is also available and well-behaved for degenerate pre-symplectic vector spaces, see Remark
\ref{rem:degenerate}. Physically, these degeneracies might be interpreted as charge observables and are worth being
studied in detail for the important examples of gauge field theories, see \cite{DHS} for the Maxwell field case.
In order to quantize fermionic gauge field theories we have to require analogously to
Definition \ref{def:matterpositive} positivity of the inner product.
\begin{defi}\label{def:gaugenondegpositive}
A fermionic classical gauge field theory $\big(M,V,W,P,K,T\big)$ is called of {\textit{positive type}} if
$\tau$ is positive definite, i.e.~$(\EE,\tau)$ is a (real) pre-Hilbert space.
\end{defi}
\sk

Bosonic classical gauge field theories can be quantized via the CCR-representation
 (see Definition \ref{def:CCR} with a possible extension to pre-symplectic vector spaces as in Remark \ref{rem:degenerate})
  and fermionic classical gauge field theories of positive type via the CAR-representation (see Definition \ref{def:CAR}). Although a quantization of fermionic classical gauge field theories in terms of a field polynomial algebra is also mathematically possible if they are not of positive type, the physical interpretation of such a quantum theory would remain unclear, cf. Remark \ref{rem:positive}.
It thus remains to study if a given fermionic classical gauge field theory $\big(M,V,W,P,K,T\big)$ is of positive type or not.
From the physical perspective it is also interesting to understand if a given bosonic classical gauge field theory
has a weakly non-degenerate $\tau$ or not.

Irrespective of the non-degeneracy or positivity of $\tau$ we can
already prove an important structural property of the (classical and
quantized) gauge field theory
corresponding to $\big(M,V,W,P,K,T\big)$. 
\begin{propo} 
Every gauge field theory $\big(M,V,W,P,K,T\big)$
  satisfies the time-slice axiom: 
  Let $\Sigma$ be an arbitrary Cauchy
  surface in $(M,g)$, and let $\Sigma_{\pm}$ be any two other Cauchy
  surfaces such that $\Sigma\subsetneq \bigl({J^-(\Sigma_+) \cap J^+(\Sigma_-)}\bigr)$. Then for every $[f]\in \EE$ there is a
  representative $f\in \Ker_0(K^\dagger)$ with $\supp (f)\subset\bigl({J^-(\Sigma_+) \cap J^+(\Sigma_-)}\bigr)$.
\end{propo}
\begin{proof}
We can obtain such $f$ by a standard construction. Let $h\in [f]$ be
arbitrary. Without loss of generality we can assume that $\supp (h)
\subset J^+(\Sigma_-)$. We pick a smooth function $\chi$ such that
$\chi=0$ on $J^-(\Sigma_-)$ and $\chi=1$ on $J^+(\Sigma_+)$ and define
\begin{flalign}
f:= h - P \chi G^{\widetilde P}_- h~.
\end{flalign}
One can now verify that $\chi G^{\widetilde P}_- h$ has compact
support, whence $[f]=[h]$, and that
$f$ has the required support property.
\end{proof}

%%%%%%%%%%%%%%%%%%%%%%%%%%%%%%%%%%%%%%%%%%%%%%%%
%%%%%%%%%%%%%%%%%%%%%%%%%%%%%%%%%%%%%%%%%%%%%%%%

\section{\label{sec:degeneracy}Non-degeneracy and positivity of gauge field theories}
Let $\big(M,V,W,P,K,T\big)$ be a classical gauge field theory and denote by 
$(\EE,\tau)$ the vector space of Proposition \ref{propo:gaugevectorspace} equipped 
with the bilinear map $\tau$, which is antisymmetric for bosonic and symmetric for fermionic theories.
In order to investigate if $\tau$ is weakly non-degenerate for bosonic or respectively positive definite for
fermionic theories it is in some cases convenient to induce an equivalent bilinear map on the space of solutions of $P$.

Let us denote by $\Sol:= \big\{\psi\in\Gamma^\infty_\sc(V): P\psi =0\big\}$ the space of
all solutions of $P$ with spacelike compact support. For every $\psi$ there exists a compact set
$C\subseteq M$, such that $\supp(\psi) \subseteq J(C)$. We can split $\psi =\psi^+ + \psi^-$ such that
 $\supp(\psi^\pm) \subseteq J^\pm(C)$. This splitting is not unique and the difference between two such 
 splittings $\psi = \psi^+ + \psi^- = \psi^{\widetilde{+}} + \psi^{\widetilde{-}}$ is given by a compactly 
 supported section $\psi^{\widetilde{+}}-\psi^+ =\psi^- - \psi^{\widetilde{-}}=: \chi \in \Gamma^\infty_0(V)$.
 We define on $\Sol$ the bilinear map
 \begin{flalign}\label{eqn:solbimap}
 \ip{~}{~}_\Sol^{}:\Sol\times\Sol \to\bbR ~,~~(\psi_1,\psi_2) \mapsto \ip{\psi_1}{\psi_2}_\Sol^{} = \ip{P\psi_1^+}{\psi_2}_{\Gamma(V)}^{} ~.
 \end{flalign}
 This map is well-defined, since firstly 
 from $P\psi_1 =0$ it follows that $P\psi_1^+ = -P\psi_1^-$ and in particular that $P\psi_1^\pm$ has compact support,
 such that the integral exists.
 Secondly,
 it is independent of the splitting,
 \begin{flalign}
 \nn \ip{P\psi_1^{\widetilde{+}}}{\psi_2}_{\Gamma(V)}^{} &= \ip{P\psi_1^{+}}{\psi_2}_{\Gamma(V)}^{} + \ip{P\chi}{\psi_2}_{\Gamma(V)}^{}\\
 &= \ip{\psi_1}{\psi_2}_\Sol + \ip{\chi}{P\psi_2}_{\Gamma(V)}^{} = \ip{\psi_1}{\psi_2}_\Sol^{} ~,
 \end{flalign}
 where we have used that $P$ is formally self-adjoint and that $P\psi_2=0$. Notice that the map (\ref{eqn:solbimap})
 is not trivial, since $\psi_1^+$ and $\psi_2$ i.g.~{\it do not} have compact overlapping support and thus we 
 {\it can not} integrate by parts $P$ to the right side.
\begin{propo}\label{propo:solspacepairing}
The following statements hold true:
\begin{itemize}
\item[(i)] The map $\ip{~}{~}_\Sol^{}$ is antisymmetric for bosonic and symmetric for fermionic gauge field theories.
\item[(ii)] The map $\ip{~}{~}_\Sol^{}$ is $\widehat{\mathcal{G}}$-gauge invariant, 
i.e.~for all $\psi\in \Sol$ and $\epsilon \in\Gamma^\infty(W)$ 
such that $K\epsilon\in\Gamma^\infty_\sc(V)$ we have $\ip{\psi}{K\epsilon}^{}_\Sol = \ip{K\epsilon}{\psi}_\Sol^{}=0$.

In particular, the map $\ip{~}{~}_\Sol^{}$ induces well-defined bilinear maps on the quotients $\Sol/\widehat{\mathcal{G}}$ and 
$\Sol/\mathcal{G}$ (remember that $\mathcal{G}\subseteq \widehat{\mathcal{G}}$).

\item[(iii)] Let $f,h\in \Ker_0(K^\dagger)$, then 
\begin{flalign}
\ip{G^{\widetilde{P}}f}{G^{\widetilde{P}}h}_{\Sol}^{} = \tau\big([f],[h]\big)~.
\end{flalign} 
\end{itemize}
\end{propo}
\begin{proof}
Proof of (i): Let $\psi_1,\psi_2\in\Sol$ be arbitrary and consider the splittings $\psi_i = \psi_i^+ + \psi_i^-$, $i=1,2$.
Notice that from $P\psi_i =0$ is follows that $P\psi_i^+ = -P\psi_i^-$.
Then
\begin{flalign}
\nn \ip{\psi_1}{\psi_2}_\Sol^{} &= \ip{P\psi_1^+}{\psi_2}_{\Gamma(V)}^{}
= \ip{P\psi_1^+}{\psi_2^+}_{\Gamma(V)}^{} +  \ip{P\psi_1^+}{\psi_2^-}_{\Gamma(V)}^{}\\
\nn &=- \ip{P\psi_1^-}{\psi_2^+}_{\Gamma(V)}^{} +  \ip{\psi_1^+}{P\psi_2^-}_{\Gamma(V)}^{}
=- \ip{\psi_1^-}{P\psi_2^+}_{\Gamma(V)}^{} -  \ip{\psi_1^+}{P\psi_2^+}_{\Gamma(V)}^{}\\
\label{eqn:tmp1} &=- \ip{\psi_1}{P\psi_2^+}_{\Gamma(V)}^{} = \mp\,\ip{P\psi_2^+}{\psi_1}_{\Gamma(V)}^{} =\mp\, \ip{\psi_2}{\psi_1}_\Sol^{}~,
\end{flalign}
where $-$ is for bosonic and $+$ for fermionic theories. All integrations by parts of $P$ in the calculation above
are well-defined, since the integrals are always over functions with compact support.
\sk

Proof of (ii):  Let $\psi\in \Sol$ and $K\epsilon \in\widehat{\mathcal{G}}$. We obtain
\begin{flalign}
\ip{\psi}{K\epsilon}_\Sol = \ip{P\psi^+}{K\epsilon}_{\Gamma(V)}^{} = \ip{K^\dagger P\psi^+}{\epsilon}_{\Gamma(W)}^{} = 0~.
\end{flalign}
In the second equality we have used that $P\psi^+$ is of compact support and in the third equality that
$K^\dagger\circ P=0 $. By (\ref{eqn:tmp1}) we have $\ip{K\epsilon}{\psi}_\Sol^{} = - \ip{K\epsilon}{P\psi^+}_{\Gamma(V)}^{}
= -\ip{\epsilon}{K^\dagger P\psi^+}_{\Gamma(W)}^{}=0$.
\sk

Proof of (iii):  Let $f,h\in\Ker_0(K^\dagger)$. Then $G^{\widetilde{P}}f,G^{\widetilde{P}}h\in\Sol$, since 
\begin{flalign}
PG^{\widetilde{P}} f = (\widetilde{P}-TK^\dagger) G^{\widetilde{P}}f = -TK^\dagger G^{\widetilde{P}}f  =0~,
\end{flalign}
where in the last equality we have used Theorem \ref{theo:gaugeproperties} (ii). The same applies for $G^{\widetilde{P}}h$.
A convenient decomposition is given by $G^{\widetilde{P}}f = G_+^{\widetilde{P}}f - G_-^{\widetilde{P}}f  $
and we find
\begin{flalign}
\ip{G^{\widetilde{P}} f }{G^{\widetilde{P}} h}_\Sol^{} = \ip{PG_+^{\widetilde{P}} f }{G^{\widetilde{P}} h}_{\Gamma(V)}^{}
=\ip{\widetilde{P}G_+^{\widetilde{P}} f }{G^{\widetilde{P}} h}_{\Gamma(V)}^{} =\ip{ f }{G^{\widetilde{P}} h}_{\Gamma(V)}^{}
=\tau\big([f],[h]\big)~,
\end{flalign}
where in the second equality we have used Theorem \ref{theo:gaugeproperties} (ii).
\end{proof}
We combine the statements proven in Theorem \ref{theo:gaugeproperties} and Proposition
 \ref{propo:solspacepairing} in order to construct an isomorphism between
 the space $(\EE,\tau)$ of Proposition \ref{propo:gaugevectorspace} and the space
 $(\Sol/\mathcal{G},\ip{~}{~}_\Sol^{})$. 
 \begin{theo}
 The sequence of maps
 \begin{flalign}
 \Ker_0(K^\dagger) \stackrel{G^{\widetilde{P}}}{\longrightarrow} \Sol \stackrel{\id}{\longrightarrow}\Sol~
 \end{flalign}
 induces a well-defined sequence of maps on the quotients (which we denote with a slight abuse of notation
 by the same symbols)
 \begin{flalign}\label{eqn:sequencegauge}
 \EE =\Ker_0(K^\dagger)/P[\Gamma^\infty_0(V)] \stackrel{G^{\widetilde{P}}}{\longrightarrow} \Sol/\mathcal{G}
 \stackrel{\id }{\longrightarrow} \Sol/\widehat{\mathcal{G}}~.
 \end{flalign}
 The first map is an isomorphism and the second map is a surjection which becomes an isomorphism if and only if
 $\mathcal{G}=\widehat{\mathcal{G}}$. Furthermore, the sequence of maps (\ref{eqn:sequencegauge}) preserves the bilinear mappings
 in $(\EE,\tau)$, $(\Sol/\mathcal{G},\ip{~}{~}_\Sol^{})$ and $(\Sol/\widehat{\mathcal{G}},\ip{~}{~}_\Sol^{})$.
 \end{theo}
 \begin{proof}
From Theorem \ref{theo:gaugeproperties} (vi) it follows that the first map is well-defined and injective.
Surjectivity of the first map follows
from Theorem \ref{theo:gaugeproperties} (iv) and (v). The second map is well-defined and surjective since
$\mathcal{G}\subseteq \widehat{\mathcal{G}}$. It is an isomorphism if and only if $\mathcal{G}=\widehat{\mathcal{G}}$.
The bilinear mappings are preserved due to Proposition \ref{propo:solspacepairing} (iii) and (ii).
 \end{proof}
 \begin{cor} \label{cor:gaugeproperties}
 If $\mathcal{G}\subset \widehat{\mathcal{G}}$ is a proper subspace then the  map $\tau$ in $(\EE,\tau)$ is degenerate.
 \end{cor}
 \begin{proof}
 Assume that $\mathcal{G}\subset \widehat{\mathcal{G}}$ is a proper subspace. Then there is a $\epsilon\in\Gamma^\infty(W)$
 such that $K\epsilon \not\in \mathcal{G}$, but $K\epsilon\in \widehat{\mathcal{G}}\subseteq \Sol$. From 
 Proposition \ref{propo:solspacepairing} (ii) we know that $\ip{\psi}{K\epsilon}_\Sol^{} =0$ for all $\psi\in \Sol$.
 Since in $\Sol/\mathcal{G}$ this $K\epsilon$ is not equivalent to zero the bilinear map
 $\ip{~}{~}_\Sol^{}$ is degenerate on $\Sol/\mathcal{G}$. Because $(\Sol/\mathcal{G},\ip{~}{~}_\Sol^{})$ is isomorphic
 to $(\EE,\tau)$ the statement follows.
 \end{proof}
\begin{rem}
This corollary might suggest that it is more convenient (regarding non-degeneracy) to 
choose $(\Sol/\widehat{\mathcal{G}},\ip{~}{~}_\Sol^{})$ instead of $(\EE,\tau)$
as the underlying vector space for a CCR or CAR-representation. There are, however, two arguments
against this choice. Firstly, the additional elements in $\widehat{\mathcal{G}}$, which are not in $\mathcal{G}$,
can not be interpreted as on-shell conditions in accord with Remark \ref{rem:formalnotation}. Secondly,
as clarified in \cite{DHS} for the Maxwell field case, the observables in $\widehat{\mathcal{G}}\setminus\mathcal{G}$
can be of physical significance.
\end{rem}
\sk

We next show that the map $\ip{~}{~}_\Sol^{}$ can be evaluated on any Cauchy surface $\Sigma \subseteq M$.
We split the globally hyperbolic spacetime $M=\Sigma^+\cup \Sigma^-$ into 
the future/past $\Sigma^\pm:=J^\pm(\Sigma)\subseteq M$ 
of the Cauchy surface $\Sigma$. We also split $\ip{~}{~}_{\Gamma(V)}^{}$, for all $f,h\in\Gamma^\infty(V)$ 
with compact overlapping support,
\begin{flalign} 
\ip{f}{h}_{\Gamma(V)}^{} = \int_{\Sigma^+} \vol\,\ip{f}{h}_V^{}
+\int_{\Sigma^-}\vol\,\ip{f}{h}_V^{} =: \ip{f}{h}_{\Gamma(V)}^+ + \ip{f}{h}_{\Gamma(V)}^-~.
\end{flalign}
This allows us to rewrite $\ip{~}{~}_\Sol^{}$ as follows, for all $\psi_1,\psi_2\in\Sol$,
\begin{flalign}
\nn\ip{\psi_1}{\psi_2}_\Sol^{} &= \ip{P\psi_1^+}{\psi_2}_{\Gamma(V)}^{} = \ip{P\psi_1^+}{\psi_2}_{\Gamma(V)}^+ + \ip{P\psi_1^+}{\psi_2}_{\Gamma(V)}^- \\
\label{eqn:cauchy}&=-\ip{P\psi_1^-}{\psi_2}_{\Gamma(V)}^+ + \ip{P\psi_1^+}{\psi_2}_{\Gamma(V)}^- ~.
\end{flalign}
In both terms we can now perform integration by parts, since the integral over the future $\Sigma^+$
 (respectively the past $\Sigma^-$)
is over a function of support in $J^-(C)$ (respectively in $J^+(C)$). The remaining boundary
 terms are then located on the Cauchy surface $\Sigma$.
 \begin{propo}
 Let $P:\Gamma^\infty(V)\to\Gamma^\infty(V)$ be a  first-order differential operator, which is formally self-adjoint
 with respect to $\ip{~}{~}_V^{}$.
 Then for all  $\psi_1,\psi_2\in\Sol$ we have for any Cauchy surface $\Sigma\subseteq M$
 \begin{flalign}\label{eqn:solprodhyper}
 \ip{\psi_1}{\psi_2}_\Sol^{} = \int_\Sigma \vols\,\ip{\sigma_P^{}(n^\flat)\psi_1\vert_\Sigma}{\psi_2\vert_\Sigma}_V^{}~,
 \end{flalign}
 where $\sigma_P^{}$ is the principal symbol of $P$, $n$ is the future pointing normal vector field of $\Sigma$,
 $\vols$ is the induced volume form on $\Sigma$ and $\vert_\Sigma^{}$ denotes the restriction of sections to $\Sigma$. 
 \end{propo}
 \begin{proof}
 This is a result of Green's formula \cite[p.~160, Prop.~9.1]{Taylor} and of $P\psi_2=0$. We have, for all $\psi_1,\psi_2\in\Sol$,
 \begin{flalign}
 \nn \ip{\psi_1}{\psi_2}_\Sol^{}  &= -\ip{P\psi_1^-}{\psi_2}_{\Gamma(V)}^+ + \ip{P\psi_1^+}{\psi_2}_{\Gamma(V)}^-\\
 \nn &=\int_\Sigma\vols\,\Big( \ip{\sigma_P^{}(n^\flat) \psi_1^-\vert_\Sigma^{}}{\psi_2\vert_\Sigma^{}}_V^{} +\ip{\sigma_P^{}(n^\flat) \psi_1^+\vert_\Sigma^{}}{\psi_2\vert_\Sigma^{}}_V^{}  \Big)\\
 &=\int_\Sigma\vols\, \ip{\sigma_P^{}(n^\flat) \psi_1\vert_\Sigma^{}}{\psi_2\vert_\Sigma^{}}_V^{}~.
 \end{flalign}
 \end{proof}

 Before we discuss our examples of  gauge field theories it is instructive to consider first the case of fermionic
  matter field theories. We will show that there are fermionic matter field theories which are not of 
  positive type (see Definition \ref{def:matterpositive}), see also \cite{Bar:2011iu}.
  This means that positivity is not a property which follows from the basic axioms of
  a fermionic classical matter or gauge field theory, see Definitions \ref{def:matterfield} and \ref{def:gaugefield}.
  \begin{ex}[Positive and non-positive fermionic matter field theories]\label{ex:positivematter}
We start with the Majorana field of Example \ref{ex:Maj1} as an example for a fermionic matter field theory of positive type.
The principal symbol of the massive Dirac operator is given by $\sigma_P(\xi) = \gamma^\mu\,\xi_\mu=\slashed{\xi}$, where 
in local coordinates $\xi = \xi_\mu dx^\mu$. The bilinear map (\ref{eqn:solprodhyper}) then reads, for all
$\psi_1,\psi_2\in\Sol$,
\begin{flalign}
\ip{\psi_1}{\psi_2}_\Sol^{} = i\,\int_\Sigma\vols\,\big(\slashed{n}\psi_1\vert^{}_\Sigma\big)^\mathrm{T} C\psi_2\vert^{}_\Sigma~.
\end{flalign}
Using Theorem \ref{theo:bernalsanchez} we obtain that the future-pointing normal vector field of the Cauchy surface $\Sigma$
is given by $n= \vartheta^{-1}\,\partial_t$, where $\vartheta$ is the positive function on $\bbR\times\Sigma$ appearing in the metric
$g=-\vartheta^2\,dt^2 \oplus g_t$ of Theorem \ref{theo:bernalsanchez}. Then $\slashed{n} = \gamma^0\vartheta=-i\beta$,
where $\beta$ is the matrix used in defining the Dirac adjoint, see Appendix \ref{app:conventions}. Since on Majorana spinors
the Dirac adjoint equals the Majorana adjoint and since $\beta^\dagger =\beta = \beta^{-1}$ we have
\begin{flalign}
\ip{\psi_1}{\psi_2}_\Sol^{} = \int_\Sigma\vols\,\psi_1^\dagger \vert^{}_\Sigma \psi_2\vert^{}_\Sigma~.
\end{flalign}
It holds that $\ip{\psi}{\psi}_\Sol^{}\geq 0$ for all $\psi\in\Sol$. Even more, $\ip{\psi}{\psi}_\Sol^{} =0$ 
implies that the initial data $\psi\vert^{}_\Sigma \equiv 0$ vanishes and thus due to the Cauchy-hyperbolicity of the 
massive  Dirac operator $\psi\equiv 0$.

An example of a fermionic matter field theory which is not of positive type is the projected Rarita-Schwinger
field presented in \cite[Section 2.6]{Bar:2011iu}. As above we use Theorem \ref{theo:bernalsanchez} to get a particularly
simple expression for the normal vector field.
 We take $V:=DM_\bbR\otimes T^\ast M$, but restrict ourselves
to the image of the projection operator defined by, 
for all $\psi\in \Gamma^\infty(V)$, $(\pi\psi)_\mu := \psi_\mu-\frac{1}{D}\gamma_\mu \gamma^\nu\psi_\nu$.
These sections satisfy $\gamma^\mu\psi_\mu=0$.
We equip the  bundle $V$ with the non-degenerate antisymmetric bilinear form 
$\ip{f}{h}_{V}^{} = i\, f_\mu^\mathrm{T} C h^\mu$. The projected Rarita-Schwinger operator
is defined by, for all $\psi\in\Gamma^\infty(V)$ with $\gamma^\mu\psi_\mu=0$, 
$(P\psi)_\mu := \slashed{\nabla}\psi_\mu - \frac{2}{D} \gamma_\mu\nabla^\nu\psi_\nu$.
It satisfies $\gamma^\mu (P\psi)_\mu=0$ for all $\psi\in\Gamma^\infty(V)$ with $\gamma^\mu\psi_\mu=0$ and thus is
a differential operator on the projected Rarita-Schwinger bundle. It is formally self-adjoint with respect to $\ip{~}{~}_V^{}$
on sections of the projected Rarita-Schwinger bundle.
The bilinear map (\ref{eqn:solprodhyper})  reads, for all $\psi_1,\psi_2\in\Sol$,
\begin{flalign}\label{eqn:projRStemp}
\ip{\psi_1}{\psi_2}_\Sol^{} = \int_\Sigma \vols \,\psi_{1\mu}^\dagger\vert_\Sigma^{}\,\psi_2^\mu\vert_\Sigma^{}~.
\end{flalign}
We can solve the constraint $\gamma^\mu\psi_\mu=0$ for $\psi_0$ and find $\psi_0 = -\gamma_0\gamma^i\psi_i$,
where $i=1,\dots,D-1$ is a spatial index.  Putting this into (\ref{eqn:projRStemp}) and setting $\psi_1=\psi_2=\psi\in\Sol$ 
leads to
\begin{flalign}
\ip{\psi}{\psi}_\Sol^{} = \int_\Sigma \vols \,\Big(\psi_{i}^\dagger\vert_\Sigma^{}\,\psi^i\vert_\Sigma^{} - (\gamma^i\psi_{i})^\dagger\vert_\Sigma^{}\,(\gamma^j\psi_{j})\vert_\Sigma^{} \Big)~.
\end{flalign}
This is an indefinite inner product, since if we evaluate it on initial data $\psi_i\vert_\Sigma^{}$
 with $\gamma^i\psi_i\vert_\Sigma^{}=0$ we obtain a positive number, while evaluating 
 it on $\psi_i\vert_\Sigma^{}=(\gamma_i \chi)\vert_\Sigma^{}$ with
  $\chi\vert_\Sigma^{}\in\Gamma_\sc^\infty(DM_\bbR)\vert_\Sigma^{}$
  we obtain a negative one.
   \end{ex}
  \sk\sk
  
We next will briefly comment on the question of weak non-degeneracy for our examples of  bosonic gauge field theories.
\begin{ex}[Linearised Yang-Mills field]\label{ex:YM2}
We analyze the case of a Yang-Mills field  linearised around a vanishing background $A_0$
 and sketch the main non-degeneracy result, see \cite{DHS} for details on the $U(1)$ case. In this case, $K=\dd^{A_0} = \dd$ 
is the exterior differential and the associated (compactly supported) de Rham cohomology groups of $M$ are defined as
\begin{subequations}
\begin{flalign}
H^n(M,\g)&:= \frac{\Ker \left(\dd: \Omega^n(M,\g) \to \Omega^{n+1}(M,\g)\right)}{\Imm \left(\dd: \Omega^{n-1}(M,\g) \to \Omega^{n}(M,\g)\right)}=H^n(M,\bbR)\otimes \g~,\\
H_0^n(M,\g) &:= \frac{\Ker \left(\dd: \Omega_0^n(M,\g) \to \Omega_0^{n+1}(M,\g)\right)}{\Imm \left(\dd: \Omega_0^{n-1}(M,\g) \to \Omega_0^{n}(M,\g)\right)}=H_0^n(M,\bbR)\otimes \g~.
\end{flalign}
\end{subequations}
We first observe that 
\begin{flalign}
\tau([f],[h])=\langle f, G^{\widetilde P}
h\rangle_{\Gamma(V)}^{}=\int_M  \ip{f}{\ast G^{\widetilde P}h}_\g^{}=0~,
\end{flalign}
 for all $f\in\Ker_0(K^\dagger)=\Ker_0(\delta)$, implies in
particular that 
\begin{flalign}
\int_M \ip{k}{\ast \dd G^{\widetilde P} h}_\g^{}=0~,
\end{flalign}
for all $k \in \Omega^2_0(M,\g)$. From the non-degeneracy of $\int_M \ip{\,\cdot\,}{\ast \,\cdot\,}_\g^{}$ 
we then obtain $\dd G^{\widetilde P} h=0$, such that $G^{\widetilde P}h$ defines an element in 
$H^1(M,\g)$. The Hodge-dual $\ast f$ for $f\in \Ker_0(\delta)$ defines an element
in $H_0^{D-1}(M,\g)$ and thus $\tau([f],[h])=0$ for all $f$ implies that
$G^{\widetilde P} h$ corresponds to the trivial element in $H^1(M,\g)$
by Poincar\'e duality (see e.g.~\cite{Bott:1995}), i.e.~$G^{\widetilde P}h = \dd \epsilon$
for some $\epsilon \in C^\infty(M,\g)$. 
This in turn implies that the necessary condition for weak non-degeneracy found in Corollary
\ref{cor:gaugeproperties} is sufficient in the case at hand.

In particular, for any spacetime with compact Cauchy surfaces we have $\mathcal{G}=\widehat{\mathcal{G}}$
and thus for the linearised Yang-Mills field with $A_0=0$ the space $(\EE,\tau)$ is symplectic.
We next provide a simple example of a spacetime for which $\mathcal{G}\subset \widehat{\mathcal{G}}$ is a proper subset, thus
$\tau$ is degenerate by Corollary \ref{cor:gaugeproperties}. Let us take Minkowski space $\bbR^D$ with flat metric $g$ and
remove the light cone of the origin $0\in\bbR^D$, i.e.~we consider the globally hyperbolic spacetime
 $M := \bbR^{D}\setminus J(\{0\})$ with the induced metric. We further take two closed balls (with strictly positive radius)
 $B_1\subset B_2\subset \bbR^D$ centred at $0$ in $\bbR^D$ and denote $B_1^M := B_1\cap M$ and $B_2^M := B_2\cap M$.
 Let us now take a function $\epsilon \in C^\infty(M,\mathfrak{g})$ such that $0\neq \epsilon = w\in \mathfrak{g}$ is a constant
 on $J(B^M_1)$ and  $\epsilon =0$ on $M\setminus J(B_2^M)$. The differential $\dd \epsilon$ is then an element
 in $\Omega_\sc^1(M,\mathfrak{g})$ and thus $\dd\epsilon \in\widehat{\mathcal{G}}$. It remains to show that there is 
 no $\tilde\epsilon\in C^\infty_\sc(M,\mathfrak{g})$ such that $\dd \epsilon =\dd\tilde\epsilon$. 
 In order to show this, let us consider the smooth embedding $\iota : (0,\infty) \to M\subset \bbR^D$ given 
 in Cartesian coordinates on $M\subset \bbR^D$ by $x\mapsto \iota(x) = (0,x,0,\dots,0)$. Pulling back the
 one-form $\dd\epsilon$ and integrating over $(0,\infty)$ we find by Stokes theorem
 \begin{flalign}
 \int_{(0,\infty)} \iota^\ast(\dd \epsilon) = \int_{(0,\infty)} \dd\iota^\ast(\epsilon)  =-w\neq 0~,
 \end{flalign}
 while doing the same for $\dd\tilde \epsilon$ with $\tilde\epsilon \in C^\infty_\sc(M,\mathfrak{g})$
 results in $0$. Thus, $\mathcal{G}\subset \widehat{\mathcal{G}}$ is a proper subset for the model under consideration
 and $\tau$ in $(\EE,\tau)$ is degenerate. For a  physical interpretation of this degeneracy
 we refer to \cite{DHS}.
\end{ex}

\begin{ex}[Linearised general relativity]\label{ex:gravity2}
If the globally hyperbolic spacetime $M$ has  compact Cauchy surfaces the weak non-degeneracy
of the pre-symplectic structure for linearised general relativity on Einstein manifolds has been shown
by Fewster and Hunt \cite[Theorem 4.3]{Fewster:2012bj}. The analysis of the non-compact case
is to our best knowledge not yet completely understood. 
\end{ex}
\sk

As it has been argued above, the positivity of a fermionic gauge field theory according to Definition
 \ref{def:matterpositive} is a physically and mathematically motivated condition. We will study this aspect for our two
 examples of fermionic gauge field theories in detail.
\begin{ex}[Toy model: Fermionic gauge field]\label{ex:toy2}
We give a simple proof that the fermionic toy model introduced in Example \ref{ex:toy1} is not of positive type.
For this proof we do not need the expression of $\tau$ on a Cauchy surface (\ref{eqn:cauchy}),
but we will work with $\tau$ as given in (\ref{eqn:sigma}). Our strategy is as follows: We assume the existence
of a $f\in\Ker_0(K^\dagger)$ such that $\tau\big([f],[f]\big) >0$ and then explicitly construct 
an $f^\prime \in \Ker_0(K^\dagger)$ such that $\tau\big([f^\prime],[f^\prime]\big)<0$.
For this we choose a basis of the symplectic vector space $\big(\bbR^{2m},\Omega\big)$, such that $\Omega$ takes
the standard form
\begin{flalign}
\Omega = \begin{pmatrix}
0 & 1 & 0 & 0 & \dots\\
-1 & 0 & 0 & 0 &\\
0 & 0 & 0 & 1 &\\
0 & 0 & -1 & 0 & \\
\vdots & & & & \ddots 
\end{pmatrix}~.
\end{flalign}
We further consider the $2m\times 2m$-matrix
\begin{flalign}
B = \begin{pmatrix}
0 & 1 & 0 & 0 & \dots\\
1 & 0 & 0 & 0 &\\
0 & 0 & 0 & 1 &\\
0 & 0 & 1 & 0 & \\
\vdots & & & & \ddots 
\end{pmatrix}~.
\end{flalign}
Let now $f\in\Ker_0(K^\dagger)$ be such that $\tau\big([f],[f]\big) >0$. Then defining
$f^\prime := Bf$ we have $f^\prime \in\Ker_0(K^\dagger)$, since $K^\dagger= \delta$ and $B$
commutes. Using that $B^\mathrm{T}\Omega B = -\Omega$ and also that $B$ commutes with $G^{\widetilde{P}}$
and the Hodge operator, we obtain
\begin{flalign}
\tau\big([f^\prime],[f^\prime]\big) = \int_M {f^\prime}^\mathrm{T}\wedge \Omega \ast G^{\widetilde{P}}f^\prime
= \int_M {f}^\mathrm{T}\wedge B^\mathrm{T}\Omega B \ast G^{\widetilde{P}}f = -\tau\big([f],[f]\big)<0~.
\end{flalign}
\end{ex}

%%%%%%%%%%%%%%%%%%%%%%%%%%%%%%%%%%%%%%%%%%%%%%%%
%%%%%%%%%%%%%%%%%%%%%%%%%%%%%%%%%%%%%%%%%%%%%%%%

\section{\label{sec:degeneracyrarita}Positivity of the Rarita-Schwinger gauge field}

We derive a sufficient condition for the positivity of the Rarita-Schwinger gauge field and
prove that this condition is satisfied on a large class of spacetimes.
\begin{theo} \label{theo:RSpos}
Consider the Rarita-Schwinger gauge field $(M,V,W,P,K,T)$ defined in Example \ref{ex:RS1}. 
Then the following statements hold:
\begin{itemize}
\item[(i)] For all $f_1,f_2\in \Ker_0(K^\dagger)$ and on a Cauchy surface $\Sigma$ as in Theorem \ref{theo:bernalsanchez}
\begin{flalign}\label{eqn:RSipdegtmp}
\tau([f_1],[f_2])=\int_\Sigma
\vols \left(\psi_{1\mu}^\dagger\vert_\Sigma^{} \,\psi_2^\mu\vert_\Sigma^{} - \frac{1}{D-2} \slashed{\psi_{1}}^\dagger\vert_\Sigma^{} \,\slashed{\psi_{2}}\vert_\Sigma^{}\right)~,
\end{flalign}
where $\psi_i := G^{\widetilde P}f_i \in\Sol$, $i=1,2$, and $\slashed{\psi} := \gamma^\mu\psi_\mu$.
\item[(ii)] Let us assume that for all $\psi \in \Sol$ satisfying $\gamma^\mu \psi_\mu=0$ there exists an $\epsilon\in \Gamma^\infty(W)$ such that
\begin{subequations}\label{eq:RSgaugeeq}
 \begin{flalign}\label{eq:hyperbolic}
\slashed{\nabla}\epsilon&=0~\qquad \quad ~~\text{on}~M~,\\
 \label{eq:elliptic}
 \gamma^i \nabla_i \epsilon &= - \gamma^i \psi_i\qquad  \text{on}~
 \Sigma~,
  \end{flalign}
  \end{subequations}
and $\epsilon\vert^{}_\Sigma$ is vanishing on the (possibly empty) boundary of $\Sigma$, whereas $\nabla \epsilon|_\Sigma$ is bounded. Then $(\EE, \tau)$ is a pre-Hilbert space, i.e.~$\tau$ is positive definite.
\item[(iii)] Let $D\ge4$ and let $M$ be asymptotically flat in the
  following sense \cite{Parker:1981uy}: There is a $t\in\bbR$, such that in a canonical 
  foliation given by Theorem \ref{theo:bernalsanchez} the Cauchy surface $(\Sigma,g_t)$
  is complete. Further, there is a compact set $C\subset \Sigma$, such that $\Sigma \setminus C$ is the disjoint union of a finite
   number of subsets $\Sigma_1, \ldots, \Sigma_N$ of $\Sigma$, each diffeomorphic to the complement of a contractible compact
   set in $\bbR^{D-1}$. Under this diffeomorphism, the Riemannian metric $g_t$ on $\Sigma_b$, $b=1,\dots,N$,
    should be of the form      
    \begin{flalign}
    \left(g_t\right)_{ij}= \delta_{ij} + a_{ij}
    \end{flalign}
in Cartesian coordinates $x^i$ of $\bbR^{D-1}$, where $a_{ij}=O(r^{-D+3})$, $\partial_k a_{ij}=O(r^{-D+2})$, and $\partial_l\partial_k a_{ij}=O(r^{-D+1})$. Furthermore, the second fundamental form (extrinsic curvature) $h_{ij}$ of $\{t\}\times \Sigma$ should satisfy $h_{ij}=O(r^{-D+2})$, $\partial_k h_{ij}=O(r^{-D+1})$. 

In this case $(\EE, \tau)$ is a pre-Hilbert space.

\item[(iv)] Let $M$ contain compact Cauchy surfaces.
In a canonical foliation given by Theorem \ref{theo:bernalsanchez} let  there be a $t\in\bbR$, such that
the induced Dirac operator on $\{t\}\times \Sigma$ has a trivial kernel.

In this case $(\EE,\tau)$ is a pre-Hilbert space.
\end{itemize}
\end{theo}
\begin{proof} Proof of (i):  The principal symbol of the Rarita-Schwinger operator \eqref{eqn:RSEOM}
reads $\sigma_P(\xi)_\mu^{~\nu} = \slashed{\xi} \delta_\mu^\nu - \gamma_\mu \xi^\nu$. Hence,
$\widetilde{(\sigma_P(n^\flat)\psi)}_\mu =\slashed{n} \psi_\mu + \frac{1}{D-2}\gamma_\mu \slashed{n}\slashed{\psi}$
and by \eqref{eqn:solprodhyper} we have
\begin{flalign}
\nn \tau([f_1],[f_2]) &= \ip{\psi_1}{\psi_2}_\Sol=\int_\Sigma
\vols\,\ip{\sigma_P^{}(n^\flat)\psi_1\vert_\Sigma}{\psi_2\vert_\Sigma}_V^{}~\\
\nn &
=\int_\Sigma
\vols\left(\ip{\slashed{n} \psi_1^\mu\vert_\Sigma^{}}{\psi_{2\mu}\vert_\Sigma^{}}_W^{}-
\frac{1}{D-2}\ip{\slashed{n}\slashed{\psi_1}\vert_\Sigma^{}}{\slashed{\psi_2}\vert_\Sigma^{}}_W^{}\right)\\
&=\int_\Sigma
\vols \left(\psi_{1\mu}^\dagger\vert_\Sigma^{} \,\psi_2^\mu\vert_\Sigma^{} - \frac{1}{D-2} \slashed{\psi_{1}}^\dagger\vert_\Sigma^{} \,\slashed{\psi_{2}}\vert_\Sigma^{}\right)~,
\end{flalign}
where the last identity follows by arguments used in Example \ref{ex:positivematter}.
\sk 

Proof of (ii): We see from (\ref{eqn:RSipdegtmp}) that
positivity in particular holds if for all $\psi\in \Sol$ we can set $\gamma^\mu\psi_\mu=0$ and $\psi_0=0$ on $\Sigma$
by a suitable choice of gauge fixing (recall that in our conventions
the metric is positive definite on spacelike vectors). It is convenient to perform such a gauge fixing in two steps. 
First, let $\psi^\prime\in\Sol$ be arbitrary. Using a $\mathcal{G}$-gauge transformation $K\epsilon$ with
$\epsilon\in\Gamma^\infty_\sc(W)$, we define $\psi_\mu := \psi^\prime_\mu + (K\epsilon)_\mu =
\psi^\prime_\mu + \nabla_\mu\epsilon - \frac{1}{2}\gamma_\mu \slashed{\nabla}\epsilon$.
Demanding $\gamma^\mu \psi_\mu =0$ leads to the equation
\begin{flalign}
\slashed{\nabla}\epsilon=\frac{2}{D-2}\,\gamma^\mu \psi^\prime_\mu~,
\end{flalign}
which can be solved for $\epsilon\in\Gamma^\infty_\sc(W)$, e.g.~by imposing a trivial initial condition.
Thus, any $\psi^\prime\in\Sol$ is $\mathcal{G}$-gauge equivalent to a $\psi\in\Sol$ satisfying
$\gamma^\mu\psi_\mu=0$.
Using Proposition \ref{propo:solspacepairing} (ii) and (\ref{eqn:RSipdegtmp}) we obtain after this gauge transformation
\begin{flalign}
\tau([f_1],[f_2])= \ip{\psi_1}{\psi_2}_\Sol=\int_\Sigma
\vols \,\psi_{1\mu}^\dagger\vert_\Sigma^{} \,\psi_2^\mu\vert_\Sigma^{}~.
\end{flalign}
Given such a $\psi\in\Sol$ with $\gamma^\mu\psi_\mu=0$ we perform a second gauge transformation
to set the zero-component $\psi_0=0$ on $\Sigma$, while preserving 
the $\gamma$-trace condition $\gamma^\mu\psi_\mu=0$ on $M$. 
The $\gamma$-trace condition is preserved by the gauge transformation $K\epsilon$, $\epsilon\in\Gamma^\infty(W)$,
 if and only if $\slashed{\nabla}\epsilon =0$ on $M$. Using this and demanding that the zero component
 of the gauge transformed section vanishes leads us to the equation \eqref{eq:elliptic}.
We assume that a solution $\epsilon\in\Gamma^\infty(W)$ of \eqref{eq:RSgaugeeq} exists, for all $\psi\in \Sol$ with 
$\gamma^\mu \psi_\mu=0$, and that $\epsilon\vert_\Sigma^{}$ is vanishing on $\partial \Sigma$ 
whereas $\nabla \epsilon\vert_\Sigma^{}$ is bounded. 

Notice that we do not demand that $\epsilon$ is an element in $\Gamma^\infty_\sc(W)$, nor that
$K\epsilon \in\Gamma^\infty_\sc(V)$. It thus remains to show that the inner product $\ip{~}{~}_\Sol^{}$
is also gauge invariant under such extended gauge transformations, more precisely that (note that 
$\widetilde{\nabla_\mu\epsilon}= \widetilde{\nabla_\mu\epsilon}^{-1}=\nabla_\mu\epsilon$ due to \eqref{eq:hyperbolic})
 \begin{flalign}\label{eq:vanish}
 \int_\Sigma
\vols\,\ip{\sigma_P(n^\flat)^{\mu\nu}\psi_{\nu}\vert_\Sigma^{}}{(\nabla_\mu \epsilon)\vert_\Sigma^{}}_W^{}
\end{flalign}
vanishes for all $\epsilon\in\Gamma^\infty(W)$ which vanish at $\partial \Sigma$ and all $\psi\in\Gamma^\infty(V)$
 which are bounded on $\Sigma$ and satisfy $P\psi=0$. To this avail, 
we note that the covariant derivative $\nabla^\Sigma$ compatible with the Riemannian metric $g_t$ on $\Sigma$ 
and $\nabla$ compatible with $g$ are related by \cite[Lemma 10.2.1]{Wald:1984}
\begin{flalign}\label{eqn:Wald}
\nabla^\Sigma_\rho T^{\alpha_1\cdots \alpha_k}_{\beta_1\cdots \beta_l}= \Pi^{\alpha_1}_{\mu_1}\cdots
 \Pi^{\alpha_k}_{\mu_k}\Pi^{\nu_1}_{\beta_1}\cdots \Pi^{\nu_l}_{\beta_l}\Pi^\lambda_\rho \nabla_\lambda T^{\mu_1\cdots \mu_k}_{\nu_1\cdots \nu_l}~,
\end{flalign}
where $\Pi^\mu_\nu := \delta^\mu_\nu + n^\mu n_\nu$ is the projector to the tangent bundle on $\Sigma$.
Since $n_\mu \sigma_P(n^\flat)^{\mu\nu} =0$ we have $\Pi_\rho^\mu \sigma_P(n^\flat)^{\rho\nu} = \sigma_P(n^\flat)^{\mu\nu}$
and we can replace $\nabla$ in \eqref{eq:vanish} by $\nabla^\Sigma$. Integration by parts 
is well-defined under the assumptions on $\psi$ and $\epsilon$. 
Using again (\ref{eqn:Wald}) in order to replace $\nabla^\Sigma$ by $\nabla$
and projectors $\Pi_\mu^\nu$, the statement follows by applying the Leibniz rule and using the equation
of motion $P\psi=0$. 
\sk

Proof of (iii): The first equation \eqref{eq:hyperbolic} for $\epsilon$ can be solved for arbitrary 
initial conditions $\epsilon|_\Sigma$ as $\slashed{\nabla}$ is Cauchy-hyperbolic, while the second equation  \eqref{eq:elliptic}
 is an elliptic constraint equation for such initial conditions, whose solvability in general depends on the topology of
$\Sigma$ and the properties of $g_t$. We shall now use
a generalisation of \cite[Theorem 4.2]{Parker:1981uy} to prove this
solvability under our hypotheses. Let $R\ge 1$
be large enough such that each $\Sigma_b\subset \bbR^{D-1}$ (we omit the
diffeomorphisms $\Sigma_b \to \bbR^{D-1}\setminus \tilde C$, with suitable contractible compact $\tilde C\subset \bbR^{D-1}$,
 here and in the following) contains the exterior of the ball $B_R$ of radius $R$. For each $b$ and each $r\ge
R$, we set $\Sigma_{b,r}:=\Sigma_b\setminus B_r$ and fix a smooth
function $\rho$ on $\Sigma$ such that $\rho\ge 1$, $\rho=r$ in
$\Sigma_{b,2R}$ and $\rho=1$ in $\Sigma\setminus \big(\bigcup_{b=1}^N
\Sigma_{b,R}\big)$. Let now $s\in \{0,1\}$ and let 
$\Vert \epsilon\Vert_{s,\delta,p}$, $\epsilon \in \Gamma_\sc^\infty(W)\vert_\Sigma$, denote the weighted Sobolev norm
\begin{subequations}
\begin{flalign}
\Vert \epsilon\Vert_{s,\delta,p}:=s \Vert\rho^{\delta+1} \nabla^\Sigma
\epsilon\Vert_p+\Vert\rho^{\delta} \epsilon\Vert_p~,
\end{flalign}
where $\nabla^\Sigma$ is the spin connection on $\Sigma$ and
\begin{flalign}
\Vert \epsilon\Vert_p := \left(\int_\Sigma \vols\, \big(\epsilon^\dagger
    \epsilon\big)^{p/2}\right)^{1/p}~.
\end{flalign}
\end{subequations}
By $\HH_{s,\delta,p}$ we denote the completion of $\Gamma^\infty_\sc(W)|_\Sigma$ with respect to 
$\Vert \cdot\Vert_{s,\delta,p}$. Let us first consider the case $D=4$.
By \cite[Theorem 4.2]{Parker:1981uy}, the map
\begin{flalign}
\gamma^i \nabla_i=:\DD: \HH_{1,\delta,p}\to \HH_{0,\delta+1,p}
\end{flalign}
is an isomorphism with a bounded inverse $\DD^{-1}$, if $p=2$, $\delta = -1$ or $p\ge 2$,
$0<\delta<2-3/p$.  Furthermore, $\DD^{-1}$ maps sections  in $\HH_{0,\delta+1,p}\cap \Gamma^\infty(W)$
to sections in $\HH_{1,\delta,p}\cap \Gamma^\infty(W)$. This proves that
\eqref{eq:elliptic} has a unique solution and that $\epsilon\in\Gamma^\infty(W)$. 
The required decay/boundedness properties of
$\epsilon\vert_\Sigma^{}$ and $\nabla\epsilon\vert_\Sigma^{}$ follow by
the arguments used in the proof of \cite[Proposition
5.3]{Parker:1981uy}.
This implies that the condition in (ii) is fulfilled and thus $(\EE,\tau)$ is a pre-Hilbert space for the
asymptotically flat case in $D=4$.

 One can straightforwardly generalise
\cite[Theorem 4.2]{Parker:1981uy} to the case $D> 4$ by noting that the part of the proof of the said theorem which is concerned with the invertability of $\DD$ for $p=2$, $\delta = -1$ can be straightforwardly generalised to $D>4$ as all inbetween steps are still valid in higher dimensions and with the steeper decay of $a_{ij}$, $\partial_k a_{ij}$ and $h_{ij}$. At the same time, these parameters are
sufficient to guarantee the required decay/boundedness properties of $\epsilon\vert_\Sigma^{}$
and $\nabla \epsilon\vert_\Sigma^{}$ for $D>4$. Hence, the condition in (ii) is fulfilled and $(\EE,\tau)$ is a pre-Hilbert space for the
asymptotically flat case in general $D\geq 4$.
\sk

Proof of (iv): The elliptic differential operator $\gamma^i\nabla_i$ on $\Sigma$ is formally skew-adjoint with
 respect to the inner product $\ip{\psi}{\chi}=\int_\Sigma \vols\, \psi^\dagger\,\chi$, see \cite[Section 3]{Parker:1981uy}
 and note the different Clifford algebra conventions used by the authors.
Thus, the trivial kernel of $\gamma^i\nabla_i$ implies a trivial kernel of its formal adjoint, 
and the solvability of (\ref{eq:elliptic}) for all source terms is guaranteed by
 the general theory of elliptic operators on vector bundles over compact Riemannian manifolds,
see e.g.~\cite[Chapter III]{Lawson} or Donaldson's lecture notes \cite[Section 3]{Donaldson}. 
Elliptic regularity implies that $\epsilon\vert_\Sigma^{} \in\Gamma^\infty(W)\vert_\Sigma^{}$.
This $\epsilon\vert_\Sigma^{}$ can be used as initial condition for solving (\ref{eq:hyperbolic}) and 
the resulting section $\epsilon\in\Gamma^\infty(W)$ satisfies the required properties, since $\Sigma$ is compact. 
Hence, the condition in (ii) is fulfilled and $(\EE,\tau)$ is a pre-Hilbert space.
\end{proof}

To close, we present an example of a Ricci-flat globally hyperbolic spacetime $M$ 
with spin structure on which the Rarita-Schwinger gauge field 
is not of positive type. Let us take $M=\bbR\times \mathbb{T}^{D-1}$, with $\mathbb{T}^{D-1}$ denoting the 
$D{-}1$-torus, equipped with the flat metric $g= -dt^2 + \sum_{i=1}^{D-1} d\varphi_i^2$. Here
$t\in \bbR$ denotes time and $\varphi_i\in [0,2\pi)$ are the angles on the torus. We choose the trivial spin structure
on $M$, in particular there exists a global basis of $\Gamma^\infty(V)$.
The equation of motion for the Rarita-Schwinger gauge field (\ref{eqn:RSEOM}) reads
$(P\psi)_\mu = \gamma^\nu\partial_\nu \psi_\mu -\gamma_\mu \partial^\nu\psi_\nu =0$.
Notice that, in particular, all constant sections $\psi_\mu \equiv \mathrm{const}$ solve this equation
and thus belong to the space $\Sol$. We obtain for such sections
\begin{flalign}
\ip{\psi}{\psi}_\Sol^{} = (2\pi)^{D-1}\,
\left(\psi^\dagger_\mu \psi^\mu - \frac{1}{D-2} \slashed{\psi}^\dagger\,\slashed{\psi}\right)~,
\end{flalign}
where $(2\pi)^{D-1}$ is the volume of the torus.
Choosing $\psi_\mu\neq 0$ such that $\psi_0 =0$ and $\gamma^i\psi_i=0$ we  obtain that
$\ip{\psi}{\psi}_\Sol^{}= (2\pi)^{D-1} \,\psi_i^\dagger\psi^i>0$. On the other hand, choosing 
$\psi_i=0$ and $\psi_0\neq 0$ we obtain
\begin{flalign}
\ip{\psi}{\psi}_\Sol^{} = -  (2\pi)^{D-1}\,\frac{D-1}{D-2}\,\psi_0^\dagger\psi_0^{} <0~.
\end{flalign}
We note that if we equip  $M=\bbR\times\mathbb{T}^{D-1}$ with one of the $2^{D-1}-1$ non-trivial spin structures
\cite{BarDirac}, the induced Dirac operator on the torus $\mathbb{T}^{D-1}$ has a trivial kernel.
Thus, the Rarita-Schwinger gauge field is of positive type by Theorem \ref{theo:RSpos} (iv).
This shows an interesting correlation between the choice of spin structure and the positivity of
the Rarita-Schwinger gauge field.

%%%%%%%%%%%%%%%%%%%%%%%%%%%%%%%%%%%%%%%%%%%%%%%%
%%%%%%%%%%%%%%%%%%%%%%%%%%%%%%%%%%%%%%%%%%%%%%%%

\begin{acknowledgements}
We would like to thank Claudio Dappiaggi, Klaus Fredenhagen, Hanno
Gottschalk, Katarzyna Rejzner, Ko Sanders, Christoph Stephan and Christoph F.~Uhlemann
for useful discussions and comments. T.P.H.~gratefully acknowledges
financial support from the Hamburg research cluster LEXI ``Connecting
Particles with the Cosmos''.
\end{acknowledgements}

%%%%%%%%%%%%%%%%%%%%%%%%%%%%%%%%%%%%%%%%%%%%%%%%
%%%%%%%%%%%%%%%%%%%%%%%%%%%%%%%%%%%%%%%%%%%%%%%%

\appendix

\section{\label{app:conventions}Spinor and gamma-matrix conventions}
We review some aspects of spinors in higher dimensions following \cite{VanProeyen:1999ni}, being mainly interested
in properties of Majorana spinors.
Let $D\text{~mod~} 8=2,3,4$ and we denote by $\eta^{ab} = \mathrm{diag}\left(-,+,+,\dots,+\right)^{ab}$ the
$D$-dimensional Minkowski metric. The $\gamma$-matrices $\gamma^a$, $a=0,\dots,D-1$, 
are complex $2^{\lfloor D/2\rfloor}\times 2^{\lfloor D/2\rfloor}$-matrices satisfying the Clifford algebra relations
$\{\gamma^a,\gamma^b\}=2\,\eta^{ab}$. We take the timelike $\gamma$-matrix to be 
antihermitian ${\gamma^0}^\dagger =-\gamma^0$ and the spatial $\gamma$-matrices 
hermitian ${\gamma^i}^\dagger= \gamma^i$, for all $i=1,\dots,D-1$. We further fix $\beta := i\gamma^0$
which satisfies $\beta^\dagger =\beta$.
There exists a charge conjugation matrix $C$, which is antisymmetric, i.e.~$C^{\mathrm{T}}=-C$, 
in the dimensions we are considering, 
see Table 1 in \cite{VanProeyen:1999ni}. Further properties are $C^\dagger = C^{-1}$ and, for all $a=0,\dots,D-1$,
\begin{flalign}
{\gamma^a}^{\mathrm{T}} = - C\gamma^a C^{-1}~.
\end{flalign}
We define the charge conjugation operation on spinors $\chi\in\bbC^{2^{\lfloor D/2\rfloor}}$ by
\begin{flalign}
\chi^c := -\beta\,C^\ast\, \chi^\ast~,
\end{flalign}
where  $\ast$ denotes component-wise complex conjugation. This operation squares to the identity,
${\chi^c}^c =\chi$, for all $\chi$. A Majorana spinor is defined by the reality condition $\chi^c =\chi$
and the space of Majorana spinors is a real vector space of dimension $2^{\lfloor D/2\rfloor}$.
For every Majorana spinor $\chi$ the Dirac adjoint equals the Majorana adjoint, 
$\overline{\chi} := \chi^\dagger \beta = \chi^{\mathrm{T}} C$, and thus the hermitian structure $\overline{\chi}\lambda$ 
on Dirac spinors equivalently reads for Majorana spinors
\begin{flalign}
\overline{\chi}\lambda = \chi^{\mathrm{T}}C\lambda = -\lambda^{\mathrm{T}} C\chi~,
\end{flalign}
where in the last equality we have used that $C^{\mathrm{T}} = -C$.
 We thus have a non-degenerate $\bbR$-bilinear antisymmetric map $\chi^{\mathrm{T}}C\lambda$
  on the space of Majorana spinors. However, this map takes values in the purely imaginary numbers $i\bbR$
and therefore should be rescaled by the imaginary unit in order to take values in the reals $\bbR$.

%%%%%%%%%%%%%%%%%%%%%%%%%%%%%%%%%%%%%%%%%%%%%%%%
%%%%%%%%%%%%%%%%%%%%%%%%%%%%%%%%%%%%%%%%%%%%%%%%

\end{document}